\documentclass[10pt, aps, pra,
twocolumn,
nofootinbib,
floatfix,
showpacs,
superscriptaddress
]{revtex4-2}
\usepackage[utf8]{inputenc}
\usepackage{enumitem}
\usepackage{physics,amsmath,amssymb,amsthm} 
\usepackage{dsfont}
\usepackage{hyperref}

\usepackage{graphicx}
\usepackage{bm}
\usepackage[dvipsnames]{xcolor}
\usepackage{times}
\usepackage{subfigure}
\usepackage{multirow}
\usepackage{mdframed}
\usepackage{tabularx}
\usepackage{float}

\hypersetup{
	colorlinks=true,
	linkcolor=CitingColor,
	citecolor=CitingColor,
	urlcolor=CitingColor
}

\definecolor{CitingColor}{rgb}{0,0.3,1}

\definecolor{nblue}{rgb}{0.2,0.2,0.9}

\newcommand{\bP}{\mathbb{P}}
\newcommand{\R}{\mathbb{R}}
\newcommand{\cA}{\mathcal{A}}
\newcommand{\cB}{\mathcal{B}}
\newcommand{\cC}{\mathcal{C}}
\newcommand{\cF}{\mathcal{F}}
\newcommand{\cG}{\mathcal{G}}
\newcommand{\cM}{\mathcal{M}}

\newcommand{\cK}{\mathcal{K}}

\newcommand{\Q}{\mathcal{Q}}
\newcommand{\cR}{\mathcal{R}}
\newcommand{\cS}{\mathcal{S}}
\newcommand{\cT}{\mathcal{T}}
\newcommand{\cU}{\mathcal{U}}
\newcommand{\cV}{\mathcal{V}}
\newcommand{\cX}{\mathcal{X}}
\newcommand{\cY}{\mathcal{Y}}
\newcommand{\cZ}{\mathcal{Z}}
\newcommand{\Hilbert}{\mathcal{H}}
\newcommand{\NS}{\mathcal{NS}}
\newcommand{\wtQ}{\widetilde{\Q}}
\newcommand{\EXT}{\mathrm{EXT}}
\newcommand{\id}{\mathds{1}}

\newcommand{\Hmin}{H_\mathrm{min}}
\newcommand{\Hmineps}{\Hmin^\epsilon}
\newcommand{\pguess}{p_\mathrm{guess}}
\newcommand{\Hsh}{H_\mathrm{sh}}
\newcommand{\kext}{k_\EXT}
\newcommand{\epsball}{\mathfrak{B}^\epsilon}

\newcommand{\deltah}{\delta_\mathrm{h}}
\newcommand{\DeltaExt}{\Delta_\EXT}
\newcommand{\DeltaInp}{\Delta_\mathrm{inp}}

\newcommand{\epsound}{\epsilon_\mathrm{s}}
\newcommand{\epscomp}{\epsilon_\mathrm{c}}
\newcommand{\epscw}{\epscomp^\mathrm{w}}
\newcommand{\epscz}{\epscomp^\mathrm{z}}
\newcommand{\epsext}{\epsilon_{\EXT}}

\newcommand{\wexp}{w_\mathrm{exp}}
\newcommand{\wtol}{w_\mathrm{tol}}
\newcommand{\ztol}{\eta_\mathrm{z}}
\newcommand{\CPTP}{\mathrm{CPTP}}
\newcommand{\freq}{\mathrm{freq}}
\newcommand{\tE}{\Tilde{E}}

\newcommand{\fgamma}{f_\gamma}
\newcommand{\fperp}{f_\perp}
\newcommand{\Max}{\mathrm{Max}}
\newcommand{\Min}{\mathrm{Min}}
\newcommand{\Var}{\mathrm{Var}}
\newcommand{\MinSig}{\Min_\Sigma}
\newcommand{\VarSig}{\Var_\Sigma}
\newcommand{\MinSigma}[1]{\mathrm{Min}_{\Sigma}(#1)}
\newcommand{\Sigamma}{\Sigma^\gamma}
\newcommand{\MinSigamma}{\Min_{\Sigamma}}
\newcommand{\VarSigamma}{\Var_{\Sigamma}}
\newcommand{\Clamb}{C_\lambda}

\newcommand{\Gamw}{\Gamma_\mathrm{win}}
\newcommand{\Gamzv}{\vec{\Gamma}_\mathrm{z}}
\newcommand{\Gamz}[1][]{
\if\relax \detokenize{#1}\relax
\Gamma_\mathrm{z}
\else\Gamma_\mathrm{zero, #1}
\fi}
\newcommand{\ztolv}{\vec{\eta}_\mathrm{z}}
\newcommand{\lambz}{\vec{\lambda}_\mathrm{z}}
\newcommand{\lambwv}{\vec{\lambda}_\mathrm{win}}
\newcommand{\lambw}[1][]{
\if\relax \detokenize{#1}\relax
\lambda_\mathrm{win}
\else\lambda_\mathrm{win, #1}
\fi}
\newcommand{\Clwz}{C_{\lambw, \lambz}}

\newcommand{\gamw}{\gamma_\mathrm{win}}
\newcommand{\gamzv}{\vec{\gamma}_\mathrm{z}}
\newcommand{\gamz}[1][]{
\if\relax \detokenize{#1}\relax
\gamma_\mathrm{z}
\else\gamma_\mathrm{zero, #1}
\fi}

\newcommand{\bracket}[1]{\langle #1\rangle}
\newcommand{\proj}[1]{|#1\rangle\langle #1|}
\newcommand{\set}[1]{\{#1\}}
\newcommand{\bigset}[1]{\bigl\{#1\bigr\}}
\newcommand{\setof}[2]{\set{#1}_{#2}}
\newcommand{\floor}[1]{\lfloor#1\rfloor}

\newcommand{\commutat}[2]{\left[#1, #2\right]}
\newcommand{\Finquad}[1]{F_{\mathrm{#1}, i}}

\newtheorem{definition}{Definition}
\newtheorem{theorem}{Theorem}

\newtheorem{lemma}{Lemma}
\newtheorem{fact}{Fact}
\newtheorem{protocol}{Protocol}

\usepackage[capitalize]{cleveref}

\begin{document}

\title{Incorporating Zero-Probability Constraints to Device-Independent Randomness Expansion}
\author{Chun-Yu Chen}
\affiliation{Institute of Information Science, Academia Sinica, Taipei 115, Taiwan}
\affiliation{Department of Physics and Center for Quantum Frontiers of Research \& Technology (QFort), National Cheng Kung University, Tainan 701, Taiwan}

\author{Kai-Siang Chen}
\affiliation{Department of Physics and Center for Quantum Frontiers of Research \& Technology (QFort), National Cheng Kung University, Tainan 701, Taiwan}

\author{Kai-Min Chung}
\affiliation{Institute of Information Science, Academia Sinica, Taipei 115, Taiwan}

\author{Min-Hsiu Hsieh}
\affiliation{Hon Hai (Foxconn) Research Institute, Taipei, Taiwan}

\author{Yeong-Cherng Liang}
\email{ycliang@mail.ncku.edu.tw}
\affiliation{Department of Physics and Center for Quantum Frontiers of Research \& Technology (QFort), National Cheng Kung University, Tainan 701, Taiwan}
\affiliation{Physics Division, National Center for Theoretical Sciences, Taipei 106, Taiwan}

\author{Gelo Noel M. Tabia}
\email{gelonoel-tabia@gs.ncku.edu.tw}
\affiliation{Department of Physics and Center for Quantum Frontiers of Research \& Technology (QFort), National Cheng Kung University, Tainan 701, Taiwan}
\affiliation{Physics Division, National Center for Theoretical Sciences, Taipei 106, Taiwan}

\begin{abstract}
One of the distinguishing features of quantum theory is that its measurement outcomes are usually unpredictable or, equivalently, {\em random}. Moreover, this randomness is certifiable with minimal assumptions in the so-called device-independent (DI) paradigm, where a device's behavior does not need to be presupposed but can be verified through the statistics it produces. In this work, we explore various forms of randomness that are certifiable in this setting, where two users can perform two binary-outcome measurements on their shared entangled state. In this case, even though the Clauser-Horne-Shimony-Holt (CHSH) Bell-inequality violation is a pre-requisite for the generation of DI certifiable randomness,  the CHSH value alone does not generally give a tight bound on the certifiable randomness. Here, we determine the certifiable randomness when zero-probability constraints are incorporated into the task of DI randomness expansion for the standard local and global randomness and the so-called "blind" randomness. Asymptotically, we observe consistent improvements in the amount of DI certifiable randomness (of all kinds) as we increase the number zero constraints for a wide range of given CHSH Bell violations. However, if we further optimize over the allowed CHSH values, then benefits of these additional constraints over the standard CHSH-based protocol are only found in the case of global and blind randomness. In contrast, in the regimes of finite data, these zero constraints only give a slight improvement in the local randomness rate when compared with all existing protocols.
\end{abstract}
\date{\today}

\maketitle

\section{Introduction}
\label{sec:Introduction}

In quantum cryptography~\cite{GRTZ02}, the device-independent~\cite{Scarani12,Brunner2014:RMP} (DI) paradigm offers a very attractive alternative to conventional schemes as it requires only a minimal set of assumptions for its security analysis~\cite{Ekert91, Mayers98, Barrett05, Colbeck06}. Indeed, as was first made explicit by Ekert~\cite{Ekert91}, the violation of a Bell inequality~\cite{Bell64} implies that that measurement outcome could not have existed before the measurement, thus leaving nothing for the adversary to eavesdrop. For example, in the simplest two-party Bell scenario, the family of Clauser-Horne-Shimony-Holt (CHSH) Bell inequalities~\cite{Clauser69} completely characterize~\cite{Fine:PRL:1982} the set of correlations admitting a local-hidden-variable model. Thus, it is natural that the CHSH parameter plays a crucial role  in many DI cryptographic tasks including quantum key distribution~\cite{Ekert91, Masanes09, Masanes11, Arnon-Friedman19, Nadlinger22, Zhang22, Tan22} and random number generation~\cite{Pironio10, Arnon-Friedman19, Brown19, Liu21}. 

Regarding single-party (local) randomness, a maximally CHSH-violating correlation yields up to 1 bit of maximum randomness~\cite{Pironio10}. However, for two-party (global) randomness, other Bell inequalities~\cite{Acin12, Bourdoncle_2019,Wooltorton22} have been found to give better randomness generation rate. Beyond these randomness forms, Miller and Shi~\cite{Miller17Jun} introduced the concept of ``blind'' randomness,\footnote{Though in~\cite{Miller17Jun}, the authors used the term ``local'' randomness, it may cause confusion since ``local'' randomness was already used to indicate the randomness from a single party. As a result, we follow the term ``blind'' randomness used in~\cite{Metger22Oct}.} signifying the unpredictability of one party's outcome even when given the other party's input and output. Clearly, this gives a strengthened version of local randomness. At the same time, this less-explored notion of randomness allows the design of cryptographic protocols for mistrustful collaborations, enabling parties to work toward a common objective without relying on mutual trust, for instance, in the certified deletion task~\cite{Fu18}. To this end, the work by Metger \textit{et al.}~\cite{Metger22Oct} (see also~\cite[Figure 2]{Metger-arXiv22}) demonstrated that the (maximal) CHSH violation alone can only certify about 0.6-bit of blind randomness in the asymptotic limit.

Several questions naturally follow. For example, could we improve the certifiable randomness for each type if we stay within the simplest Bell scenario? For global randomness, the maximal amount of 2-bit is known~\cite{Wooltorton22} to be certifiable using a family of self-testing~\cite{Mayers98, Supic2020} quantum correlations, but not from the CHSH violation alone. After all, these self-testing correlations do not maximally violate the CHSH Bell inequality but rather a different family of Bell inequalities~\cite{Wooltorton22}. 
Indeed, more randomness can generally be certified from the same data if we use the full correlation~\cite{Nieto-Silleras:2014aa,Bancal:2014aa}  or consider several Bell estimators simultaneously~\cite{Nieto-Silleras:2018aa}. Alternatively, one can also hope to gain a better certification by imposing further constraints in addition to the CHSH Bell value. The intuition here is that if the additional constraints do not, {\em a priori}, exclude a self-testing correlation, then it may improve the certifiable randomness by restricting the possible eavesdropping strategies of the adversaries.

Coming back to Bell experiments, we remind that no-signaling conditions~\cite{PR94,BLM+05} have to be enforced for any meaningful DI randomness generation. Mathematically, these conditions define a set of correlations called the ``no-signaling'' ($\NS$) polytope. Interestingly, even though $\NS$ differs from the quantum ($\Q$) set of correlations, they share~\cite{Goh2018} some nontrivial, common boundaries, which can be characterized by the number of zeros appearing in the correlation vector and their relative positions~\cite[Table V]{CTJ+22}.
An example of such a correlation is the one exhibiting the well-known Hardy paradox~\cite{Hardy1993}. 
In general, the correlation that violates the CHSH Bell-inequality maximally within each class is even known to exhibit robust self-testing~\cite[Table VI]{CTJ+22}. Given that certain quantum correlations lying on these common boundaries are~\cite{KA:IEEE:2020} useful for the task of randomness amplification~\cite{CR:NatPhys:2012}, one may also wonder whether they can similarly provide an advantage for the task of DI randomness generation.

In this work, we study (1) the amount of DI randomness that can be extracted for a given CHSH value when the various zero-probability constraints of~\cite{CTJ+22} are incorporated, and (2) the amount of DI randomness extractable when we employ the CHSH-maximizing quantum strategy within each class.
For the former, we evaluate randomness in the asymptotic limit as von Neumann entropy through the Brown-Fawzi-Fawzi (BFF21) method~\cite{Brown21} whereas for the latter, we work in finite regimes as smooth min-entropy by applying the generalized entropy accumulation theorem (GEAT)~\cite{Metger22Oct}.
The dual variables in the semidefinite program (SDP) for the asymptotic rate computation with the BFF21 method can be used to construct the min-tradeoff function, a necessary component in any approach that involves entropy accumulation. In standard DI quantum key distribution or a DI randomness expansion protocol, testing rounds are used to estimate the Bell value and this serves as the spot-checking method embedded in the protocol. Here, we follow~\cite{Liu21} to convert the min-tradeoff function into a ``crossover'' min-tradeoff function~\cite{Dupuis18} that accounts for entropy that can be accumulated from both testing rounds and non-testing rounds. 
By meticulously evaluating the required quantities of the min-tradeoff function, we can establish bounds on the finite rate for various types of randomness.


\section{Preliminaries}
\label{sec:Preliminaries}

\subsection{Security definition}
\label{subsec:Security}

Two parameters are relevant for describing the security of a randomness expansion protocol: soundness and completeness.

A protocol is said to be $\epsound$-sound if the final state $\rho_{KE}=\sum_{k\in\cK} \proj{k}\otimes\rho^k_E$ of the output system $K$ and the adversary's side information $E$ (including both quantum and classical parts) follows
\begin{equation}
\Pr[\mathrm{NonAbort}]\frac{1}{2}\|\rho_{KE|\mathrm{NonAbort}}-\mu_K\otimes\rho_E\|_1 \le \epsound,
\label{eqn:soundness_def}
\end{equation}    
where $\Pr[\mathrm{NonAbort}]$ is the probability of the protocol is not aborted after checking the termination criteria, $\rho_{KE|\mathrm{NonAbort}}$ is the normalized state conditioned on the non-aborting event, $\mu_K=\frac{1}{|\cK|}\id$ is the maximally mixed state on $\Hilbert_K$, and $\|\sigma\|_1=\Tr\sqrt{\sigma^\dagger \sigma}$ is the Schatten 1-norm.

A protocol is said to have $\epscomp$-complete if, in the honest implementation, the aborting probability is bounded
\begin{equation}
\Pr[\mathrm{Abort}|\mathrm{Honest}] \le \epscomp.
\label{eqn:completeness_def}
\end{equation}

\subsection{Entropic quantities}
\label{sebsec:entropic quantities}
We define the different entropic quantities that are useful to quantify the amount of randomness extracted from the source $K$ against given side information $E$.

The first entropic quantity is conditional von Neumann entropy, which describes the averaged extractable randomness or in other words the extractable randomness in the asymptotic limit. For a bipartite state $\rho_{KE}$ of the source $K$ and side information $E$ on $\Hilbert_K\otimes\Hilbert_E$, the conditional von Neumann entropy $H(K|E)_{\rho}$ on $\rho_{KE}$ is defined as
$$H(K|E)_{\rho}:=H(\rho_{KE})-H(\rho_E).$$

The other entropic quantity we consider in our work is the smooth min-entropy $\Hmineps(K|E)_{\rho}$.
Before defining this useful quantity, let's first introduce another meaningful quantity, min-entropy $\Hmin(K|E)_{\rho}$, that quantifies the extractable randomness in the worst case.
$$\Hmin(K|E)_{\rho}:=\sup_{\sigma_E} \sup\set{\lambda\in\R \mid \rho_{KE}\le e^{-\lambda}\id_K\otimes\sigma_E},$$
for any sub-normalized state $\rho_{KE}	\preceq 1$. Especially, when $\rho_{KE}$ is a classical-quantum (cq) state, i.e., $\rho_{KE}=\sum_{k\in\cK} \proj{k}_K\otimes\rho^k_E$, with the classical variable $K\in\cK$, min-entropy can be interpreted in terms of guessing probability
$$\Hmin(K|E)_\rho =-\log \pguess(K|E)_\rho,$$
where the guessing probability is defined as~\cite{Tomamichel17, Brown19}
$$\pguess(K|E)_\rho = \sup_{E^k} \sum_{k\in\cK} \Pr(K=k)_\rho \Tr[E^k\rho_{E|K=k}],$$
where $\Pr(K=k)_\rho=\Tr[\rho^k_E]$ is the probability of measuring $k$ on the system $K$, $\rho_{E|K=k}$ is the normalized version of $\rho^k_E$, and $\set{E^k}_k$ is a set of positive operator-valued measures (POVMs) acting on system $E$ and $\rho_{E|K=k}$ is the normalized post-measurement state after performing the measurement on $K$ and getting outcome $k$. 

Associated with the min-entropy definition, the smooth min-entropy can be defined as~\footnote{This definition follows Tomamichel's book~\cite{Tomamichel15}; however, other papers may use a different distance metric, e.g., in~\cite{Renner05}, they use trace distance $\Delta(\rho, \sigma)=\frac{1}{2}\|\rho-\sigma\|_1 + \frac{1}{2}|\Tr(\rho-\sigma)|$ to define the $\epsilon$-ball, which is also fine, since we can bound the radius of the ball with different metrics by the inequality $\Delta(\rho,\sigma)\le D(\rho,\sigma) \le \sqrt{2\Delta(\rho,\sigma)}$}
$$\Hmineps(K|E)_\rho := \max_{\tilde{\rho}_{KE}\in\epsball(\rho_{KE})} \Hmin(K|E)_{\tilde{\rho}},$$
where $\epsball(\rho_{KE}) = \set{\tilde{\rho}_{KE}\mid D(\tilde{\rho}_{KE}, \rho_{KE})\le \epsilon}$, and $D(\rho,\sigma) = \sqrt{1-\left( \Tr|\sqrt{\rho}\sqrt{\sigma}| + \sqrt{(1-\Tr\rho)(1-\Tr\sigma)} \right)^2}$ is the purified distance. Note that by taking the limit $\epsilon\to 0$, we have $\lim_{\epsilon\to 0}\Hmineps(K|E)_{\rho} = \Hmin(K|E)_{\rho}$.
The $\epsilon$ parameter allows us to inspect the states around the state $\rho_{KE}$ we assumed or observed. Smooth min-entropy simply takes the best state that gives the largest worst-case randomness.

To design a practical DI randomness expansion protocol, it is necessary to confine our attention to a finite number of variables. Consider a protocol generating $n$ outputs sequentially $K^n=K_1K_2...K_n$. The associated side information becomes $E_n=I^n E'_n$, where $I^n=I_1I_2...I_n$ and $I_k$ refers to the classical side information generated in the $k$-th round, and $E'_n$ is the quantum side information after $n$ rounds.
The amount of randomness stemming from this $n$-round sequential protocol can be quantified through the smooth min-entropy, $\Hmineps(K^n|I^n E'_n){\rho_{K^nI^nE'_n}}$.

\subsection{Types of randomness}
\label{subsec:Types of randomness}

Let $\cG$ be a bipartite nonlocal game defined by the input sets $\cX\times\cY$ with $\cX=\{X\}$, $\cY=\{Y\}$, and the outcome sets $\cA\times\cB$, with $\cA=\{A\}$, $\cB=\{B\}$, all linked to a winning condition function $\omega:\cX\times\cY\times\cA\times\cB \mapsto \R$. We explore three distinct types of randomness. These types are classified by the target system $K$ where the randomness is extracted from, and the adversary's quantum side information $E'$ and classical side information $I$. The amount of randomness can be described by the entropic quantity in the following form~\footnote{We write in the form of the von Neumann entropy, however, this can be adapted to other kinds of entropic quantities with the same target system and conditional systems.}
$$H(K|IE')_{\rho_{KIE'}}.$$

\textbf{Local randomness} is a type of randomness when only one party's outcome is considered as the source for randomness generation, e.g., $K=A$, and the classical side information is $I=XY$. The relevant entropic quantity for local randomness is $H(A|XYE')$.

\textbf{Global randomness} is a type of randomness when considering randomness extracted from both parties' outcomes $K=AB$. With the same classical side information specified as in the local randomness case, i.e., $I=XY$, the relevant entropic quantity for global randomness is $H(AB|XYE')$.

The aforementioned types of randomness are the standard DI randomness, restricting the adversary from directly participating in the nonlocal game. That is the adversary is prohibited from directly accessing outputs from any of the parties. The randomness under this scenario has been proposed and discussed for a long time~\cite{Colbeck06, Pironio10, Colbeck11}.

Apart from the standard types of randomness, another type of randomness called \textbf{Blind randomness} was proposed~\cite{Miller17Jun, Fu18} to consider a different scenario where the adversary is allowed to be part of the parties in the game, i.e., has the capability to either generate or acquire the outcome from one or more than one of the parties. The randomness is computed against other parties' side information. For instance, in the CHSH game, if we consider extracting randomness from Alice's outcome $K=A$, then, Bob's outcome $B$ is assumed to be part of the adversary's classical side information $I=XYB$, and the quantum side information $E'$ also extends to the part of the bipartite state that is used to generate Bob's outcome $B$. In this case, the relevant entropic quantity is $H(A|BXYE')$. Essentially, this scenario introduces the idea of mistrust between the parties.

\subsection{Extractor}
\label{subsec:Extractor}

To complete the randomness expansion protocol, an indispensable tool is required to extract randomness from a classical source that may be correlated with quantum side information, i.e., the \textit{quantum-proof strong extractor}~\cite{De12_Trevisan}.

\begin{definition}~\cite[Lemma 3.5]{De12_Trevisan}
\label{def:quantum-proof extractor}
For any input source $X\in\cX$ in the form of a cq state $\rho_{XE}$ with quantum side information $E$ satisfying $\Hmineps(X|E)_{\rho}\ge \kext$, after applying a quantum-proof $(\kext, \epsext)$-strong extractor, $\EXT:\cX\times\cY\to\cZ$, with seed $Y\in\cY$, the extractor generates output $Z\in\cZ$. The state $\rho_{ZYE}$ satisfies
\begin{equation}
\frac{1}{2} \|\rho_{ZYE} - \mu_Z\otimes\mu_Y\otimes\rho_E\|_1 \le \epsext+2\epsilon,
\end{equation}
where $\mu_Z$ ($\mu_Y$) is the maximally mixed state on the Hilbert space $\Hilbert_Z$ ($\Hilbert_Y$).
\end{definition}

There are multiple ways to construct a quantum-proof strong extractor, a well-known construction is using \textit{two-universal hash functions} or \textit{$\deltah$-almost two-universal hash functions}~\cite{Tomamichel10}, while another construction called \textit{Trevisan's extractor}~\cite{De12_Trevisan} provides a way to build an extractor consuming a shorter seed. For completeness, we describe the quantum-proof extractor with $\deltah$-almost two-universal hash functions as an example in the Appendix~\ref{appendix:Quantum-proof extractor}.


\section{Protocol}
\label{subsec:Protocol}

The nonlocal game corresponding to the correlations from one of the no-signaling boundary (NSB) classes~\cite{CTJ+22} can be constructed by the CHSH game with zero-probability constraints. To this end, consider the CHSH nonlocal game $\cG$ with the inputs $x,y\in\set{0,1}$ chosen according the uniform probability distribution $P(x,y)=\frac{1}{4}$ for all $x,y\in\set{0,1}$, and the outputs $a,b\in\set{0,1}$. The winning condition function is defined as
\begin{equation}
    \omega(x,y,a,b)=xy\oplus a \oplus b=\begin{cases}
        1,\ &\text{win},\\
        0,\ &\text{loose}.
    \end{cases}
\end{equation}
The classes of zero-probability constrained quantum correlations~\cite{CTJ+22} are defined as

\begin{equation}
\begin{aligned}
\Q_{3a} = \set{&\vec{P}\in\Q \mid \\
&P(0,0|0,0) = P(1,1|1,0) = P(1,1|0,1) = 0},\\
\Q_{3b} = \set{&\vec{P}\in\Q \mid \\
&P(0,0|0,0) = P(1,1|0,0) = P(1,0|1,1) = 0},\\
\Q_{2a} = \set{&\vec{P}\in\Q \mid P(0,0|0,0) = P(1,1|0,0) = 0},\\
\Q_{2b} = \set{&\vec{P}\in\Q \mid P(0,0|0,0) = P(1,1|1,0) = 0},\\
\Q_{2c} = \set{&\vec{P}\in\Q \mid P(0,0|0,0) = P(1,0|1,1) = 0},\\
\Q_{1} = \set{&\vec{P}\in \Q \mid P(0,0|0,0) = 0}.
\end{aligned}
\label{eqn:zpc_classes}
\end{equation}

We herein also define the sets of input-output tuples for zero-probability constraints 
\begin{equation}
    \cS_{\kappa} = \set{(a,b,x,y)|\vec{P}\in\Q_{\kappa}, P(a,b|x,y)=0},
\end{equation}
where $\kappa$ denotes the class of the zero-probability constrained quantum correlations. For example, for class 2a, the set is $\cS_{2a}=\set{(0,0,0,0), (1,1,0,0)}$; for class 3b, the set is $\cS_{3b}=\set{(0,0,0,0), (1,1,0,0), (1,0,1,1)}$.

In practice, it is impossible to achieve zero-probability constraints due to noise and device imperfections, so we consider the relaxation of these constraints with a small tolerable error $\ztol$, i.e., $\forall\vec{P}\in\widetilde{\Q}^{\ztol}_{\kappa}$, $P(a,b|x,y)\le \ztol$ for all $(a,b,x,y)\in\cS_{\kappa}$, where $\wtQ^{\ztol}_{\kappa}$ is the relaxed version of $\Q_\kappa$.

We state the steps of DI randomness generation protocol for three types of randomness with chosen class $\kappa$ in Protocol~\ref{protocol:DI randomness generation}. Among these classes, it turns out that the best protocols for both local and global randomness in terms of the maximal asymptotic rate are given by $\kappa=\text{2a}$, while $\kappa=\text{3b}$ provides the optimal asymptotic rate for blind randomness.

To demonstrate the realization of Protocol~\ref{protocol:DI randomness generation}, we describe the quantum strategy that maximizes the CHSH winning probability for class 2a and class 3b as follows:
\begin{equation}
\begin{aligned}
\ket{\Psi} &= \cos\theta\ket{01} + \sin\theta\ket{10}, \\
A_0 &= B_0 = \sigma_z,\\
A_1 &= \cos 2\alpha \sigma_z - \sin 2\alpha \sigma_x, \\
B_1 &= \cos 2\beta \sigma_z - \sin 2\beta \sigma_x,
\end{aligned}
\end{equation}
where the parameters $\theta = \frac{\pi}{4}, \alpha = -\frac{5\pi}{6}, \beta = \frac{\pi}{6}$ for class 2a; $\alpha = \beta \approx 0.6354$, $\theta = \tan^{-1}(\tan\alpha \tan\beta)$ for class 3b (see Table~\ref{tab:QuantumRealize} or \cite[Section III]{CTJ+22}) for the realization of all the classes.)

\begin{figure}
\begin{mdframed}
\begin{protocol}{\bf DI randomness generation protocol with \\zero-probability constraints.}
\label{protocol:DI randomness generation}
\end{protocol}

\begin{enumerate}[leftmargin=2em]
    \item Alice and Bob decide which type of randomness they want to generate
    $$\chi=\begin{cases}
        \text{0, for local randomness,} \\
        \text{1, for global randomness,} \\
        \text{2, for blind randomness.}
    \end{cases}$$
    They also decide which class of zero-probability constraints, $\kappa$, they want to apply.
    \item For each round $i\in\{1,2,...,n\}$
    \begin{enumerate}[leftmargin=10pt]
        \item Alice chooses $T_i\in\{0,1\}$ with probability $P(T_i=1)=\gamma$. For $T_i=1$, Alice selects $X_i=x, Y_i=y$ according to the probability distribution $P(x,y)=\frac{1}{4}$, for all $x,y\in\set{0,1}$; if $T_i=0$, Alice sets $X_i=x^*, Y_i=y^*$.
        \item Alice inputs $X_i$ into her device and sends $Y_i$ to Bob. Alice computes her output $A_i$ and \begin{enumerate}
            \item if $\chi=2$, she requests Bob to return his output $B_i$ all the time.
            \item otherwise, she only asks Bob to return when $T_i=1$.
        \end{enumerate}
        \item If $T_i=1$, Alice computes $C_{\omega,i}=\omega(X_i, Y_i, A_i, B_i)$ and sets $C_{z,(a',b',x',y'),i}=\delta_{a',A_i}\delta_{b',B_i}\delta_{x',X_i}\delta_{y',Y_i}$
        for all $(a',b',x',y')\in\cS_{\kappa}$; if $T_i=0$, Alice sets $C_{\omega,i}=C_{z,(a',b',x',y'),i}=\perp$ for all $(a',b',x',y')\in\cS_{\kappa}$.
    \end{enumerate}
    \item After $n$ rounds, Alice checks $\abs{\set{i:C_{\omega,i}=1}} \ge (\wexp-\wtol)\gamma n$ and $\abs{\set{i:C_{z,(a', b',x',y'),i}}} \le \ztol\gamma n$ for all $(a',b',x',y')\in \cS_{\kappa}$. She aborts if any one of the checks is not satisfied.
    \item If the protocol is not aborted, Alice defines the target source $K_i$ and public information $I_i$ in each round based on $\chi$
    \begin{enumerate}
        \item if $\chi=0$, she sets $K_i=A_i$ and $I_i=X_iY_i$.
        \item if $\chi=1$, she sets $K_i=A_iB_i$ and $I_i=X_iY_i$.
        \item if $\chi=2$, she sets $K_i=A_i$ and $I_i=X_iY_iB_i$.
    \end{enumerate}
    Then, Alice computes $\kext$ such that $\Hmineps(K^n|I^nE')\ge \kext$, and applies a quantum-proof $(\kext, \epsext)$-strong randomness extractor to extract $\floor{\kext-\DeltaExt}$ bits from $\setof{K_i}{i,T_i=0}$.
\end{enumerate}
\end{mdframed}
\end{figure}


\section{Finite Analysis on Randomness}
\label{sec:Results}

The asymptotic rates for all classes are calculated using the BFF21 method~\cite{Brown21} (see also~\cref{appendix:Asymptotic rate}) with uniform input probability, i.e., $P(x,y)=\frac{1}{4}$ for all $x,y\in\set{0,1}$, tolerable winning probability deviation $\wtol=2\times 10^{-5}$,\footnote{That is even when $\wexp$ is set at the quantum bound, we still allow the winning probability in the range $[\wexp-\wtol, \wexp]$.} and tolerable error for zero probability $\ztol=10^{-10}$. In Table~\ref{tab:asym_rate_cls}, we provide the maximum asymptotic rates based on the maximal quantum-achievable winning probability for each class~\footnote{Note that the values listed in Table~\ref{tab:asym_rate_cls} are not very tight because of the tradeoff taken between numerical precision and the feasibility of the solver for the whole class. In particular, we choose the Mosek solver parameter that describes the allowable gap between the primal and dual solutions equal to $10^{-6}$. For CHSH, the local value is analytically given by~\cite{Pironio10} as 1.} (see Table~\ref{tab:QuantumRealize} for the quantum-achievable winning probability of each class.)

\begin{table}
    \centering
    \renewcommand{\arraystretch}{1.2}%
    \vspace{.5em}
    \begin{tabular}{cccc}
    \hline
     Class & Local & Global & Blind \\
    \hline\hline
     CHSH & 0.9981 & 1.5816 & 0.5823 \\
     1  & 0.9958 & 1.6770 & 0.7381 \\
     2a & 0.9992 & 1.7964 & 0.7962 \\
     2b & 0.9450 & 1.5963 & 0.7914 \\
     2b\textsubscript{swap} & 0.9977 & 1.5963 & 0.6501 \\
     2c & 0.9329 & 1.7028 & 0.7806 \\
     3a & 0.9497 & 1.5448 & 0.7444 \\
     3b & 0.9777 & 1.6852 & 0.9238 \\
     \hline
    \end{tabular}
    
    \caption{Maximal asymptotic rate (conditional von Neumann entropy, $H(K|E)$) of all types of randomness for each class with protocol parameters: $\wtol=2\times 10^{-5}$, and $\ztol=10^{-10}$ with eighteen-term Gauss-Radau quadrature. Note that 2b\textsubscript{swap} is the party-swapped version of class 2b.}
    \label{tab:asym_rate_cls}
\end{table}

The curves of the asymptotic DI randomness rate as a function of the CHSH winning probability for all classes are shown in \cref{fig:asymp_all_rand}.

\begin{figure}
\centering
\subfigure[Local randomness]{\includegraphics[width=.9\linewidth]{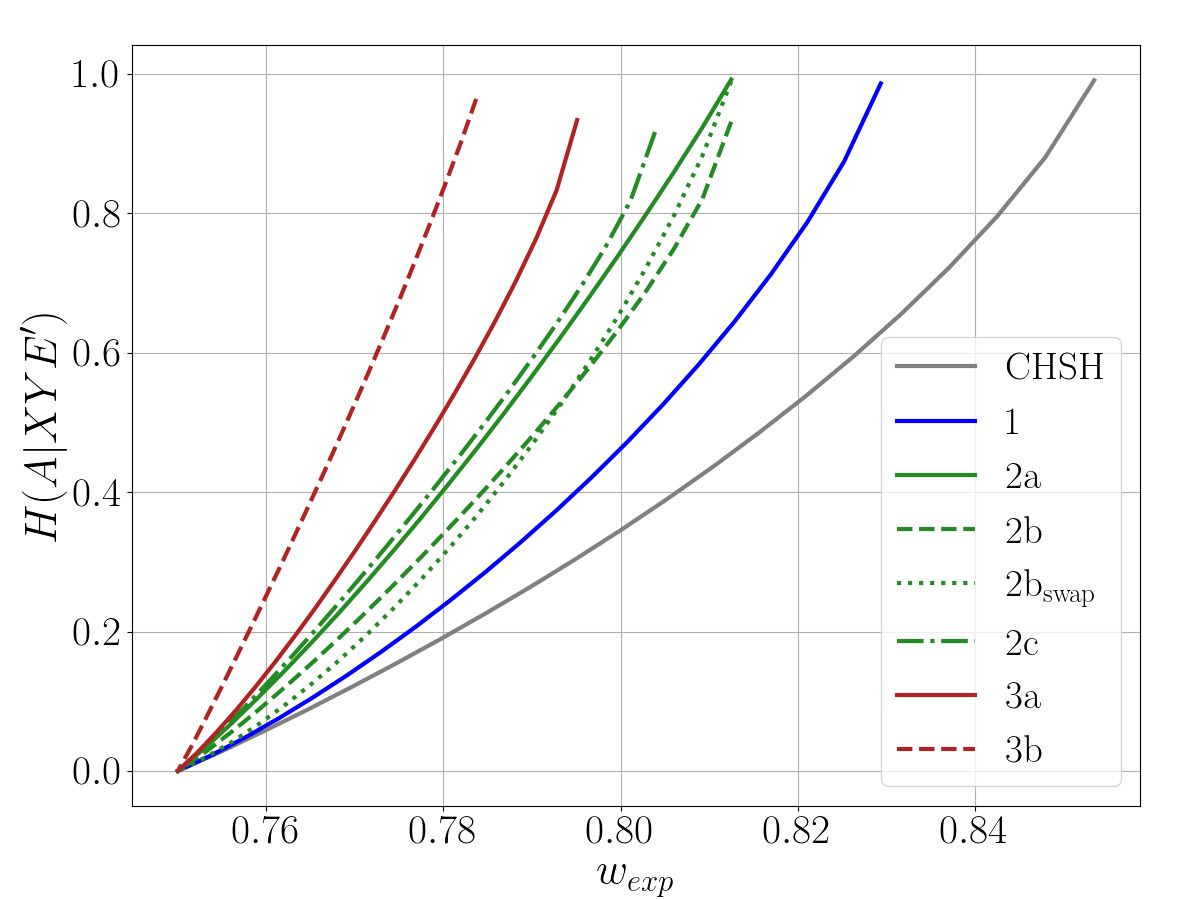}
 \label{fig:single_asymp}}\\
\subfigure[Global randomness]{\includegraphics[width=.9\linewidth]{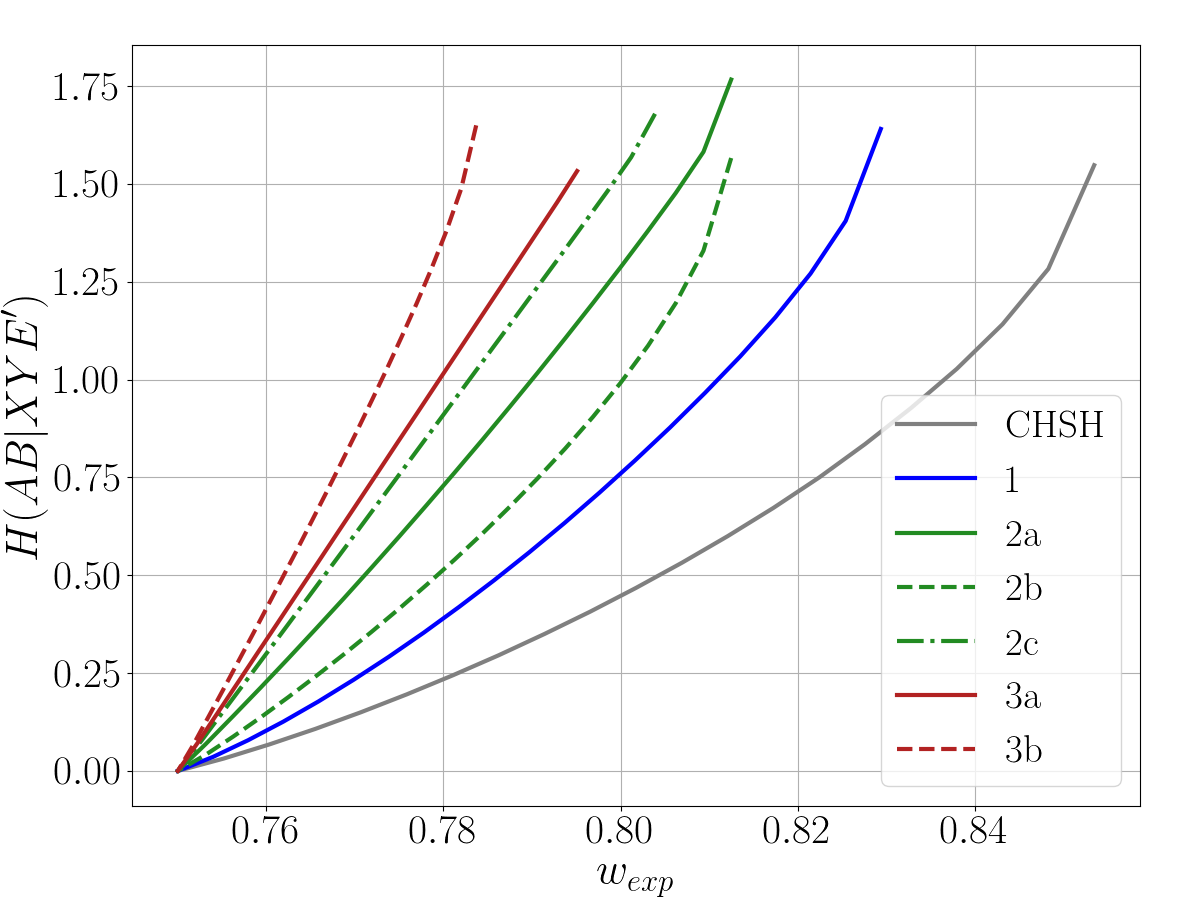}
\label{fig:global_asymp}}\\
\subfigure[Blind randomness]{\includegraphics[width=.9\linewidth]{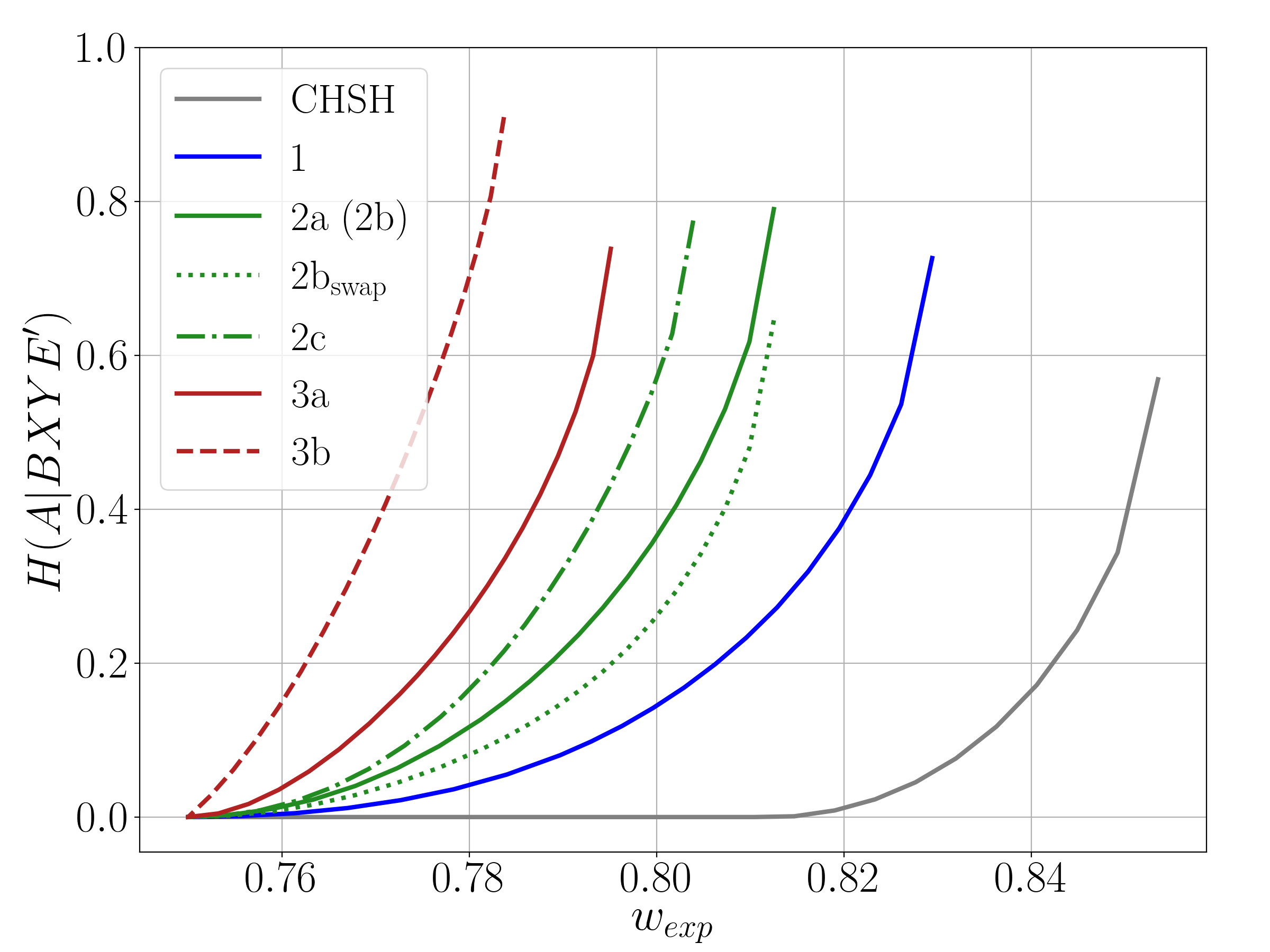}
\label{fig:blind_asymp}}
\caption{The asymptotic rate (conditional von Neumann entropy, $H(K|E)$) of (a) local (b) global and (c) blind randomness with optimal choice of the inputs over CHSH winning probability. The rates are computed in NPA level 2 plus $ABZ^{(\dagger)}, ABZ^\dagger Z, ABZZ^\dagger$ types of moments with Gauss-Radau quadrature containing $m=12$ terms.}
\label{fig:asymp_all_rand}
\end{figure}

To determine the finite rate of randomness, we utilize the GEAT (see Appendix~\ref{appendix:Entropy accumulation}.) Essentially, the GEAT establishes a lower bound on the smooth min-entropy for an $n$-round scenario. The bound on the smooth min-entropy is obtained by subtracting round-dependent correction terms from the worst-case von Neumann entropy. These correction terms rely on specific properties of the min-tradeoff function and the security parameters.

In Appendix~\ref{appendix:Construction of min-tradeoff function}, we systematically outline the construction of the min-tradeoff function, which is derived from the Lagrange dual function of the semidefinite program (SDP) for conditional von Neumann entropy (asymptotic rate) computation.

For a spot-checking embedded protocol, the min-tradeoff function needs to include results from both testing and non-testing rounds. This kind of min-tradeoff function called \textit{crossover min-tradeoff function} can be constructed from the original min-tradeoff function that only contains the testing rounds. The transformation of the min-tradeoff function into its crossover version is demonstrated in Appendix~\ref{appendix:Cross-over min-tradeoff function}.

With all the required ingredients in place, we can numerically derive the lower bounds on the smooth min-entropy for various types of randomness. Taking the input randomness consumption into account, we illustrate the finite rate of Protocol~\ref{protocol:DI randomness generation} in Fig.~\ref{fig:finite_rate} with the formula stated in Theorem~\ref{thm:fin_rate_prot1}.

\begin{theorem}[Security and finite rate of Protocol~\ref{protocol:DI randomness generation}]
\label{thm:fin_rate_prot1}
Given $\wexp\in[w_C, w_Q]$, $\gamma\in(0,1]$, $\nu'\in[0,1]$, and $\epsilon, \wtol, \ztol, \beta\in(0,1)$, Protocol~\ref{protocol:DI randomness generation} is $\epsound$-sound and $\epscomp$-complete, where $\epsound = \epsext+2\epsilon$ and $\epscomp = 1 - (1 - e^{-2\wtol^2 n})(1 - e^{-2\ztol'^2 n})^{n_\mathrm{zero}}$, $n_\mathrm{zero}$ is the number of zero-probability constraints depending on the choice of the classes $\kappa$. The finite rate is given by
\begin{multline}
r(n, \wexp, \gamma, \epsilon, \wtol, \ztol, \lambda, \beta, \nu') = \\
h(\wexp-\wtol) - \Delta - \DeltaExt - \DeltaInp,
\label{eqn:fin_rate}
\end{multline}
where $\Delta$ is finite-rate correction term, $\DeltaExt$ is the extractor entropy loss, and $\DeltaInp$ is the randomness consumption to generate inputs. $\Delta$ is given below
\begin{equation}
\begin{aligned}
\Delta =& \frac{\ln 2}{2}\beta \left[\log 9 + V\right]^2 \\
&- \frac{1}{n}\bigl[ \frac{1+\beta}{\beta} \log(1-\sqrt{1-\epsilon^2}) + \frac{1+2\beta}{\beta}\log\Pr[\Omega]\bigr] \\
&- \frac{1}{6\ln 2} \frac{\beta^2}{(1-\beta)^3} \zeta^\beta \ln^3(\zeta+e^2),
\end{aligned}
\end{equation}
where $V=\sqrt{2+D-\min_{\nu\in[1-w_Q,w_Q]}\lambda^2(\nu-\nu_0)^2}$, $\gamma_0=\frac{1-\gamma}{\gamma}$,
$\nu_0=\frac{1}{2\gamma}-\gamma_0\nu'$, $D=\frac{\lambda^2}{4\gamma^2} - \frac{\gamma_0}{\gamma} \lambda^2(1-\nu')\nu'$, $\zeta=2^{2+\lambda (\gamma_0 (1-\nu') + w_Q)}$, $w_C$ and $w_Q$ are classical and quantum bound of the game.

Eq.\eqref{eqn:extractor_entropy_loss} gives the extractor entropy loss for the extractor built based on the $\deltah$-almost hashing.

The input randomness consumption is given
\begin{equation}
\DeltaInp=\gamma H_\mathrm{sh}(p) + H_\mathrm{sh}(\gamma),
\label{eqn:inp_rand_consump}
\end{equation}
where $\Hsh$ is Shannon entropy.
\end{theorem}

The derivation of the security and finite-rate correction term in Theorem~\ref{thm:fin_rate_prot1} can be found in Appendix~\ref{appendix:Protocol security and finite rate lower bound}.

To demonstrate the security of the Protocol~\ref{protocol:DI randomness generation}, we set the smoothness $\epsilon=10^{-12}$ and the extractor parameter $\epsext \le \epsilon\times 10^{-3}$. The soundness of Protocol~\ref{protocol:DI randomness generation} is $\epsound =\epsext+2\epsilon\approx 2\times 10^{-12}$. By choosing the tolerable winning probability deviation $\wtol=10^{-4}$ and tolerable error for zero probability $\ztol=10^{-3}$, the completeness of Protocol~\ref{protocol:DI randomness generation} is $\epscomp=0.8223$ for $n=10^7$.

Since the extractor entropy loss $\DeltaExt$ is independent of the number of rounds $n$, we neglect it in our computation for convenience.
The finite rates of global randomness for Protocol~\ref{protocol:DI randomness generation} are illustrated in Fig.~\ref{fig:finite_global}. The standard CHSH and the $\delta$-family Bell inequality~\footnote{The $\delta$-family Bell is in the following form, $\bracket{A_0 B_0} + \frac{1}{\sin\delta}(\bracket{A_0 B_1} + \bracket{A_1 B_0}) - \frac{1}{\cos2\delta}\bracket{A_1 B_1}$.}~\cite{Wooltorton22} are included for comparison. All the lines are computed with the tolerable deviation $\wtol=10^{-4}$, and $\delta=\pi/6$.
We choose two different zero-probability tolerance $\ztol=10^{-3},\ 10^{-9}$ to show the robustness of Protocol~\ref{protocol:DI randomness generation}.

\begin{figure}
\centering
\subfigure[Local randomness]{\includegraphics[width=.9\linewidth]{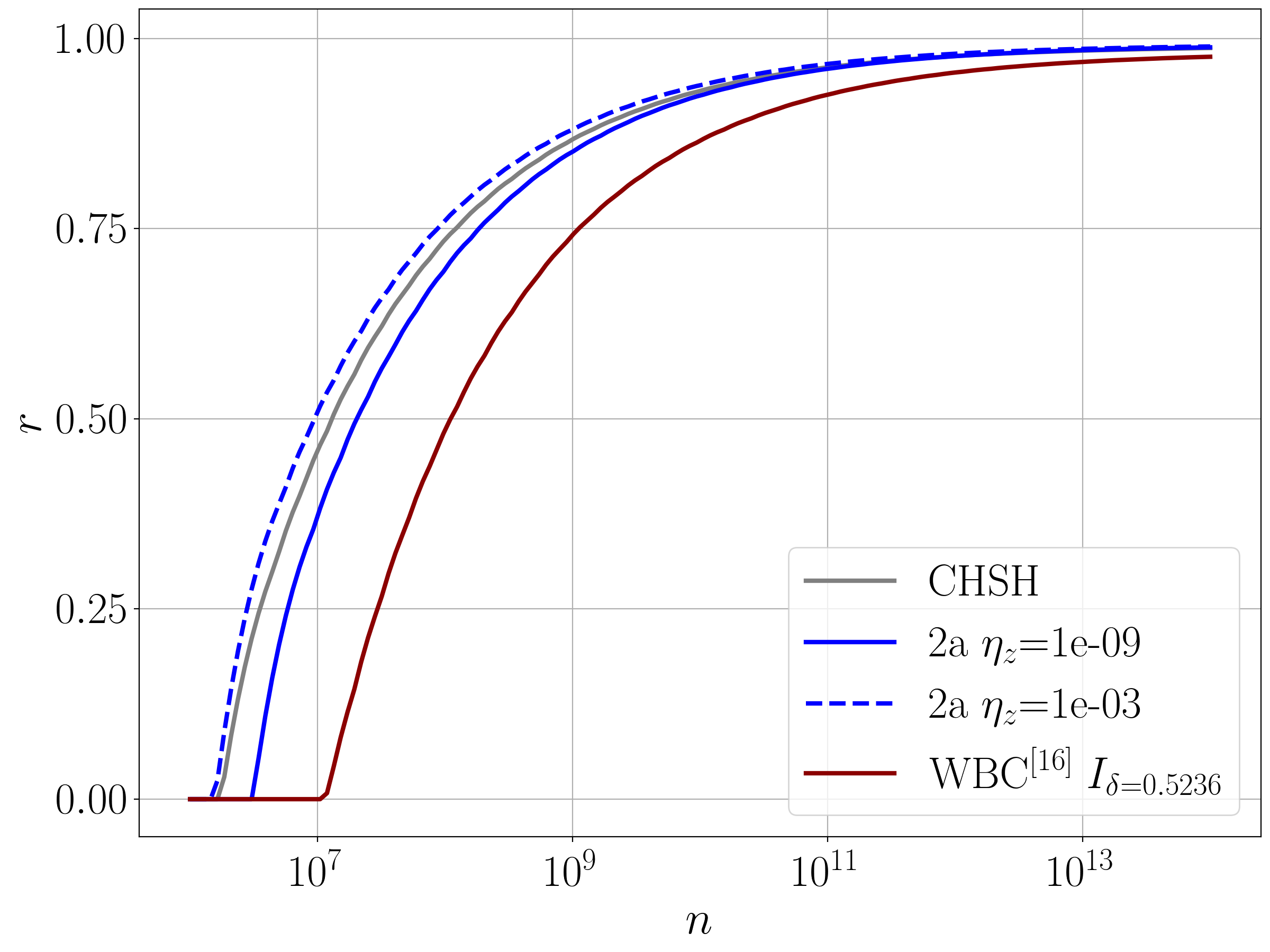}
\label{fig:finite_single}}\\
\subfigure[Global randomness]{\includegraphics[width=.9\linewidth]{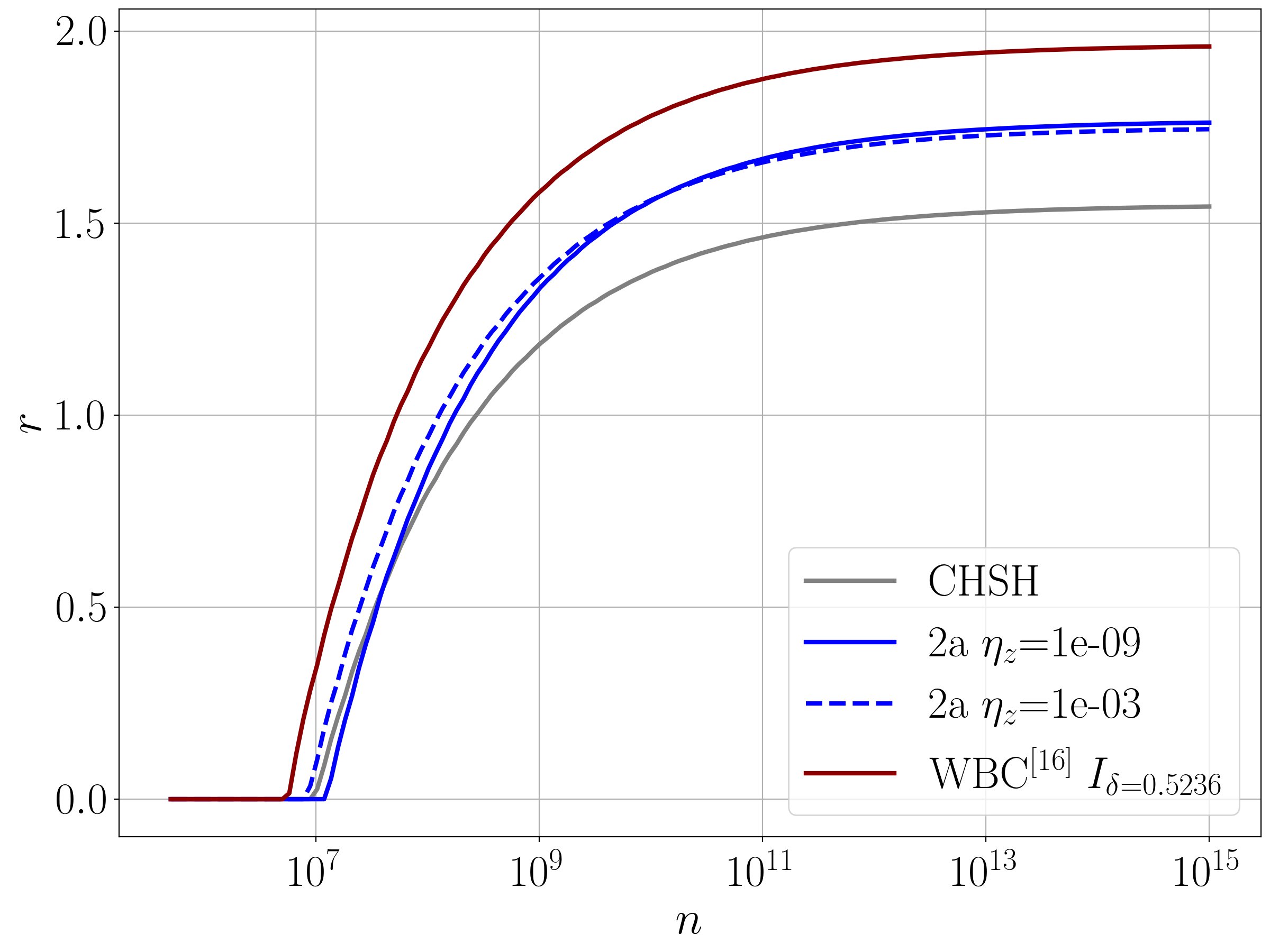}
\label{fig:finite_global}}\\
\subfigure[Blind randomness]{\includegraphics[width=.9\linewidth]{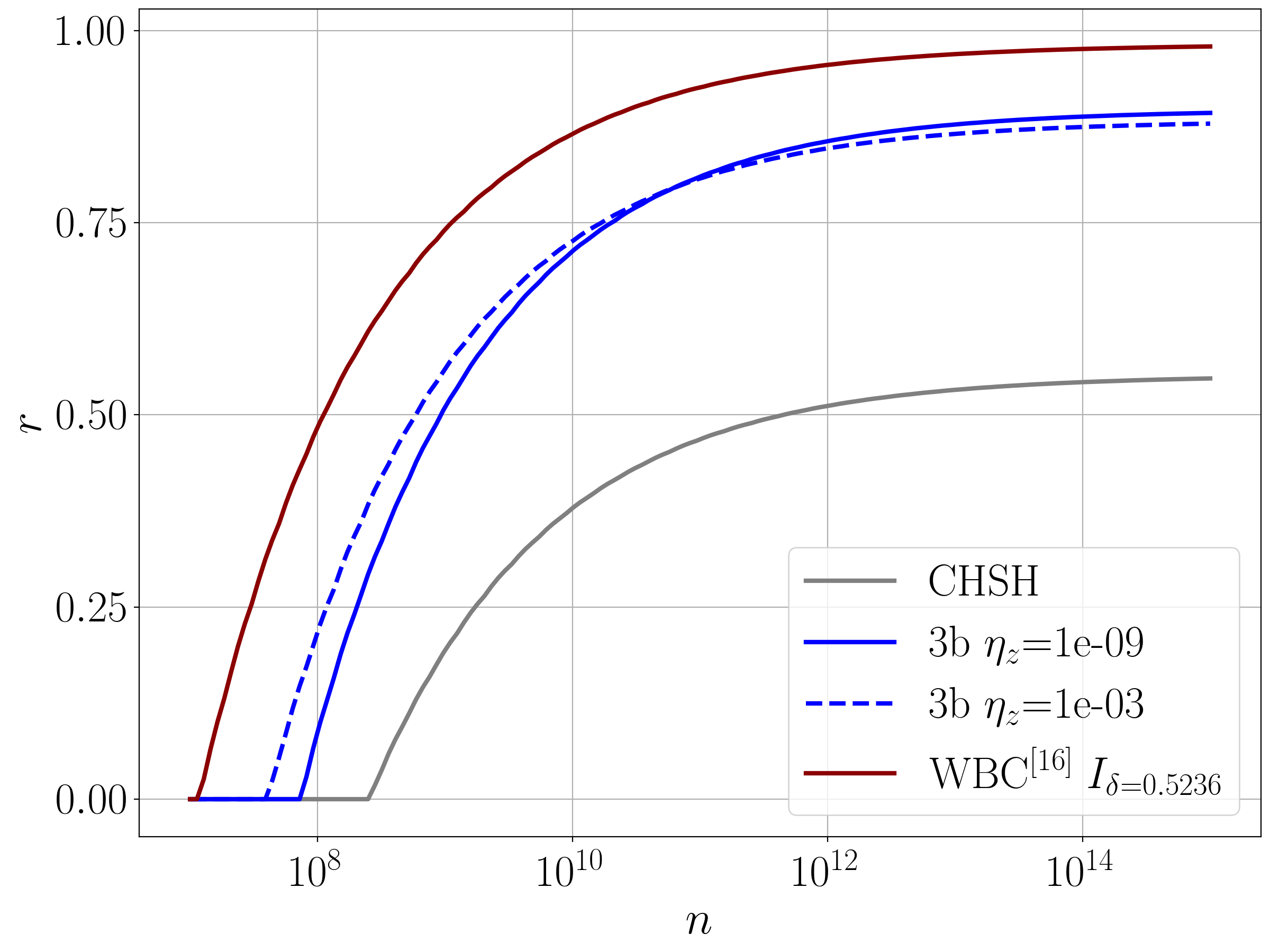}
\label{fig:finite_blind}}
\caption{Finite rates of (a) local (b) global and (c) blind randomness for Protocol~\ref{protocol:DI randomness generation} with zero-probability constraints, class 2a for (a) and (b), class 3b for (c), compared to the standard CHSH-game-based protocol and the one relying on WBC $I_\delta$ with $\delta=\pi/6$. Two different zero-probability tolerances $\ztol=10^{-3}, 10^{-9}$ and a fixed winning probability tolerance $\wtol=10^{-4}$ are considered.}
\label{fig:finite_rate}
\end{figure}

Other types of randomness and the finite analysis of other classes are shown in Appendix~\ref{appendix:Numerical results}.


\section{Concluding remarks}
\label{sec:Discussion}

The amount of DI randomness for a given CHSH score is generally not tight. In this work, we have shown how improvements in the asymptotic rates across all forms of randomness can be achieved for a wide range of CHSH scores by incorporating the zero-probability constraints classified in~\cite{CTJ+22}, see~\cref{fig:asymp_all_rand}. 
In particular, for the case of blind randomness, when relying solely on the CHSH score, the rate becomes zero when the score (winning probability) drops below 0.81. In contrast, for all the other protocols that incorporate these additional zero-probability constraints, the rate reaches zero only at the classical bound of 0.75. However, if one further optimizes the rate over the range of CHSH scores allowed, then an obvious advantage in the DI randomness generation rate persists solely for the blind and global randomness (see ~\cref{tab:asym_rate_cls}).

Next, we focus on the quantum strategy from each class that maximizes the DI randomness generation rate, via Protocol~\ref{protocol:DI randomness generation}, for all types of randomness. For completeness, we also compare the best (finite) rates achievable from these classes against those obtained from the quantum strategy maximizing the CHSH Bell-inequality violation and the $\delta$-family inequality of~\cite{Wooltorton22} with $\delta=\pi/6$. As expected from the asymptotic results, we see from~\cref{fig:finite_rate} a general advantage from a protocol with additional zero-probability constraints against the standard CHSH protocol for all three types of randomness. However, these advantages are, at the same time, inferior to those brought by the protocol~\cite{Wooltorton22} of employing the $I_{\delta}$ inequality. The only exception to this observation is for the single-party (local) randomness, where our protocol offers a slight improvement in the minimal rounds required to observe a nonzero generation rate.

Our results suggest that when there is a limitation on the number of rounds, it may be possible to explore different numbers of zero-probability constraints and various tolerable levels to attain a higher rate. However, to actually take advantage of our observation, it is clear that one should perform further analysis by taking into account also other imperfections, such as losses, that are, unfortunately, very common in photonic experiments.

\acknowledgments
We thank Stefan Bauml for helpful discussions on the entropy accumulation theorem. This work is partially supported by the National Science and Technology Council (NSTC, formerly Ministry of Science and Technology), Taiwan (Grants No. 109-2112-M006-010-MY3, 112-2119-M001-004, 112-2119-M001-006, 112-2628-M006-007-MY4).

\appendix


\section{Quantum-proof extractor based on almost two-universal hashing}
\label{appendix:Quantum-proof extractor}

\begin{definition}[$\deltah$-almost two-universal hashing~\cite{Hayashi16}]
A family $\cF:\cX\to\cY$ of hash functions is said to be \textit{$\deltah$-almost two-universal} if the following condition hold
\begin{equation}
\Pr_{f\in\cF}[f(x)=f(x')]\le\frac{\deltah}{|\cY|}\quad \forall\ x\neq x',
\end{equation}
where the probability is computed uniformly over the family $\cF$, $\deltah\ge 1$. Especially, when $\deltah=1$, the family $F$ is called \textit{two-universal} hash function family.
\end{definition}

\begin{lemma}[Quantum Leftover Hash Lemma~\cite{Tomamichel10}]\label{lem:quan_leftover}
Given any $\epsilon'>0$, a cq-state $\rho_{XE}=\sum_{x\in\cX}\proj{x}\otimes\rho^x_E$, and a $\deltah$-almost two-universal hash family $\cF:\cX\to\cZ$, where $|\cZ|=2^l$, if we apply a hash function on the classical register by uniformly choosing a function $f$ from the family $\cF$, i.e., $\rho_{FZE}=\frac{1}{|\cF|} \sum_{f\in\cF}\proj{f}_F \otimes \sum_{x\in\cX}\proj{f(x)}_Z \otimes \rho^x_E$, then we have~\footnote{In \cite{Tomamichel10}, the author uses $\frac{1}{2}\min_{\substack{\sigma_E\in\cS(\Hilbert_E) \\ \Tr\sigma_E=\Tr\rho_E}}\|\rho_{AE}-\mu_A\otimes\rho_E\|_1$ instead of the form in the LHS of Eq.\eqref{eqn:quan_leftover}. While in their proof, they first take an upper-bound on $\frac{1}{2}\min_{\substack{\sigma_E\in\cS(\Hilbert_E) \\ \Tr\sigma_E=\Tr\rho_E}}\|\rho_{AE}-\mu_A\otimes\rho_E\|_1$ as the form of the LHS of Eq.\eqref{eqn:quan_leftover}. Hence, the quantum leftover hash lemma is also compatible with this definition.}
\begin{multline}
\frac{1}{2} \|\rho_{FZE}-\mu_F\otimes\mu_Z\otimes\rho_E\|_1 \le \\
\frac{1}{2}\sqrt{\deltah-1+2^{l-\Hmin(X|E)+\log(2/\epsilon'^2+1)}}+\epsilon',
\label{eqn:quan_leftover}
\end{multline}
where $F$ is the system that records the choice of the hash function, and $\mu_F$ ($\mu_Z$) is the maximally mixed state on the Hilbert space $\Hilbert_F$ ($\Hilbert_Y$).
\end{lemma}

In Lemma~\ref{lem:quan_leftover}, the choice of a hash function from the hash family could be implemented by an extra random source, which is called the seed and can be made public. Thus we consider a so-called extractor function with an extra input as a seed, namely $\EXT:\cX\times\cY\to\cZ$, where $\cX$ denotes the set of inputs, $\cY$ denotes the set of seeds, and $\cZ$ denotes the set of outputs, and the size of the set of the seeds is equal to that of the hash function family, i.e., $|\cY|=|\cF|$.

\begin{lemma} \label{lem:quan_extr_del_alm}
A quantum-proof $(\kext,\epsext)$-strong extractor, $\EXT:\cX\times\cY\to\cZ$, where $\log|\cX|=n, \log|\cY|=d$, and $\log|\cZ|=l$, can be implemented by $\deltah$-almost two-universal hash function with the output length
\begin{equation}
l=\kext-\log(1+\frac{2}{{\epsilon'}^2})-\log\frac{1}{4{\epsilon''}^2-\deltah+1},
\label{eq:2.11}
\end{equation}
where $\epsext=\epsilon'+\epsilon''$ and $\deltah$ and $\epsilon''$ satisfy $\deltah\le 1+4\epsilon''^2$.
\end{lemma}

\begin{proof}
First, set the upper bound in Eq.\eqref{eqn:quan_leftover} as the security parameter $\epsext$
\begin{equation*}
\epsext = \frac{1}{2}\sqrt{\deltah-1+2^{l-\Hmin(X|E)+\log(2/\epsilon'^2+1)}}+\epsilon'.
\end{equation*}
Next, let $\epsilon''=\frac{1}{2}\sqrt{\deltah-1+2^{l-\Hmin(X|E)+\log(2/\epsilon'^2+1)}}$. The output length is then
\begin{equation}
l=\Hmin(X|E)-\log(1+\frac{2}{{\epsilon'}^2})-\log\frac{1}{4{\epsilon''}^2-\deltah+1}.
\label{eq:2.12}
\end{equation}
Finally, with the condition on the state $\rho_{XE}$ of input source $X$ with side information $E$, $\Hmin(X|E)\ge \kext$, we can safely substitute the lower bound $k$ into Eq.\eqref{eq:2.12}.
\end{proof}

\begin{fact}
\label{fact:entropy consumption of an extractor}
A quantum-proof extractor needs to consume entropy to provide security. For example, in the case of a $(\kext, \epsext)$-strong extractor based on $\deltah$-almost two-universal hashing, this consumption of entropy is quantified as 
\begin{equation}
\DeltaExt = \log(1+\frac{2}{\epsilon'^2}) + \log{\frac{1}{4\epsilon''^2-\deltah+1}},
\label{eqn:extractor_entropy_loss}
\end{equation}
where $\epsext=\epsilon'+\epsilon''$ and $\deltah$ and $\epsilon''$ satisfy $\deltah\le 1+4\epsilon''^2$. Given a source with entropy lower bound, $\kext$, the $(\kext, \epsext)$-extractor can extract at most $\kext-\DeltaExt$ bits from the source.
\end{fact}


\section{Realization of zero-probability constrained quantum correlations}
\label{appendix:Realization of ZPC classes}

For our purposes, it is enough to distinguish the classes according to the number and positions of the zero probabilities.

For each of the no-signaling boundary classes, the quantum correlation that maximally violates the CHSH inequality has been shown to self-test some reference states and measurements, which are provided in Table~\ref{tab:QuantumRealize}. They also provide some numerical calculations of the robustness of the self-testing property of each of them.

\begingroup
\renewcommand{\arraystretch}{1.15}
\begin{table*}[t!]
\centering
    \begin{tabular}{ccccc}
    \hline
         Classes & State &  Extra conditions & $w_\mathrm{CHSH}^\mathrm{max}$ & 
         $\vec{P}_\mathrm{CHSH}^\mathrm{max}$ \\ \hline\hline
        1  & $\ket{\Psi_1}$ & None  & $0.8294$ & $\theta = \frac{\pi}{4}, \phi \approx 0.2275, \alpha = \beta \approx -0.6403$
        \\ \hline 
        2a  & $\ket{\Psi_2}$ & None & $0.8125$ & $\theta = \frac{\pi}{4}, \alpha = -\frac{5\pi}{6}, \beta = \frac{\pi}{6}$
        \\ \hline
        2b  & $\ket{\Psi_3}$ & None  & $0.8125$ & $\phi = \beta = \frac{\pi}{4}, \alpha = \frac{\pi}{6}$
        \\ \hline
        2c  & $\ket{\Psi_1}$ & $\phi = \tan^{-1}(\frac{\sin\theta}{\tan\beta}-\tan\alpha\cos\theta)$ & $0.8039$ &
        $\theta \approx 0.5815, \alpha \approx 0.8068, \beta = \frac{\pi}{2} - \alpha$
        \\ \hline        
        3a & $\ket{\Psi_3}$ & $\phi = \tan^{-1}(\frac{\tan\beta}{\sin\alpha})$ & $0.7951$ & $\alpha = \beta = \frac{1}{2}\tan^{-1}\left(-2\sqrt{2+\sqrt{5}}\right)$
        \\ \hline
        3b & $\ket{\Psi_2}$ & $\theta = \tan^{-1}(\tan\alpha \tan\beta)$ & $0.7837$ & $\alpha = \beta \approx 0.6354$
        \\ \hline \hline 
 \multicolumn{5}{c}{
 Alice's observables: $A_0 = \sigma_z, \; A_1 = \cos 2\alpha \sigma_z - \sin 2\alpha \sigma_x$} \\
  \multicolumn{5}{c}{
 Bob's observables: $B_0 = \sigma_z, \; B_1 = \cos 2\beta \sigma_z - \sin 2\beta \sigma_x$}
 \\ \hline       
 \multicolumn{5}{c}{
 $\ket{\Psi_1} = \cos\phi(\cos\theta\ket{01} + \sin\theta\ket{10}) + \sin\phi\ket{11}$}
 \\  
 \multicolumn{5}{c}{
 $\ket{\Psi_2} = \cos\theta\ket{01} + \sin\theta\ket{10}$}
 \\
 \multicolumn{5}{c}{
 $\ket{\Psi_3} = \sin\phi(\cos\alpha\ket{01} - \sin\alpha\ket{11}) + \cos\phi\ket{10}$}
 \\ \hline 
\end{tabular}
\caption{Quantum realization of all the classes of zero-probability constrained quantum correlations in~\cite{CTJ+22}.
}
\label{tab:QuantumRealize}
\end{table*}
\endgroup


\section{Asymptotic rate computation with BFF21 method}
\label{appendix:Asymptotic rate}

In~\cite{Brown21}, Brown \textit{et al.} provided a method to numerically find a lower bound on the von Neumann entropy in the Bell scenario. This method is useful for many DI cryptography tasks, especially randomness certification. The main idea of the work is to handle the nonlinear logarithm function by Gauss-Radau quadrature approximation. Together with the famous Navascu\'es-Pironio-Ac\'{i}n (NPA) hierarchy~\cite{NPA07, NPA08}, they derived a semidefinite program (SDP) which is tractable with many commercial solvers, including Mosek~\cite{mosek}. Besides, the NPA hierarchy of the given Bell scenario can be easily constructed by Wittek's Python package, \texttt{ncpol2sdpa}\footnote{We use the version that Peter Brown currently maintains.}~\cite{ncpol2sdpa}. Here we show the SDP that gives an infimum on the asymptotic rates of randomness.

In the CHSH scenario, let Alice and Bob's measurement POVMs be $\setof{\setof{M^a_x}{a}}{x}$ and $\setof{\setof{N^b_y}{b}}{y}$ respectively. Choose the Gauss-Radau quadrature approximation with $m$ terms. We the set of non-Hermitian operators acting on the adversary's quantum system $E'$ be $\setof{Z_{ab}}{ab}$ (or $\setof{Z_{a}}{a}$ for the standard local randomness). The von Neumann entropy can be lower-bounded by

\begin{equation}
H(K|IE') \ge c_0 + \sum_{i=1}^{m} c_i \Xi_i,
\label{eqn:gaussradau_quadrature}
\end{equation}
where the coefficients $c_0=\sum_{i=1}^{m} c_i$, $c_i=\frac{w_i}{t_i\ln2}\ \ \forall i \in \set{1,2,...,m}$, $\set{t_i}^{m}_{i=1}$ and $\set{w_i}^{m}_{i=1}$ are the nodes and weights of the Gauss-Radau quadrature of the interval $(0,1]$, and $\Xi_i$ is the SDP corresponding to the $i^\mathrm{th}$ term of the quadrature~\footnote{In~\cite{Brown21}, they claimed the last term of the quadrature of which the node is at the endpoint 1 can be trivially bounded by some constant, $1/m^2\ln2$, which may reduce the tightness of the obtained bound. While if one does the optimization with $t_m=1$, the SDP sometimes can not be solved to achieve an optimal solution. We came out with a way to handle it by changing the endpoint from 1 to $0.9999$, which can provide a much tighter bound.}.

\begin{equation}
\begin{aligned}
\Xi_i = &\min_{\set{M^a_{x^*}}, \set{N^b_{y^*}}, \set{Z_{ab}}} \langle F_i(\setof{M^a_{x^*}}{a}, \setof{N^b_{y^*}}{b}, \setof{Z_{ab}}{ab}) \rangle \\
    s.t. &\sum_{abxy} c_{abxy} \langle M^a_x N^b_y\rangle \ge \wexp-\wtol \\
    &\sum_{abxy} c_{abxy} \langle M^a_x N^b_y\rangle \le \wexp+\wtol \\
    &P(ab|xy)\le \ztol\quad \forall (a,b,x,y)\in \cS_{\kappa} \\
    &\sum_a M^a_x = \id, \sum_b N^b_y = \id \quad\forall\ x,y \\
    &M^a_x \ge 0, N^b_y \ge 0 \quad\forall\ a,b,x,y \\
    &\commutat{\setof{M^a_x}{ax}}{\setof{N^b_y}{by}} = 0 \\
    &\commutat{\setof{M^a_x}{ax}}{\setof{Z^{(\dagger)}_{ab}}{ab}} = 0\\
    &\commutat{\setof{N^b_y}{by}}{\setof{Z^{(\dagger)}_{ab}}{ab}} = 0.
    \label{eqn:inner_quad_sdp}
\end{aligned}
\end{equation}
Note that the minimization is taken inside the summation in Eq.\eqref{eqn:inner_quad_sdp} and thus we need to compute $m$ SDPs for $m$-term quadrature. In principle, taking the minimization outside the summation with $\setof{\setof{Z^i_{ab}}{ab}}{i=1}^{m}$ would provide a tighter bound. However, this requires much more memory to construct the NPA hierarchy with a level higher or equal to 2. We follow the way Brown \textit{at el.}~\cite{Brown21} did to make it more feasible by exchanging the summation and minimization. So, we can reduce the required memory and still obtain a valid (but looser) lower bound simultaneously.  

For the standard local randomness, all $\setof{Z^{(\dagger)}_{ab}}{ab}$ in Eq.\eqref{eqn:inner_quad_sdp} should be replaced with $\setof{Z^{(\dagger)}_{a}}{a}$. The objective function inside the quadrature summation is
\begin{equation}
\begin{aligned}
\Finquad{local}(\setof{M^a_{x^*}}{a}, \setof{N^b_{y^*}}{b}, \setof{Z_{a}}{a}) = \Finquad{local}(\setof{M^a_{x^*}}{a}, \setof{Z_{a}}{a}) \\
= \sum_a \left[M^a_{x^*}(Z_a + Z_a^\dagger + (1-t_i)Z_a^\dagger Z_a) + t_i Z_a Z_a^\dagger\right].
\end{aligned}
\label{eqn:inner_quad_local}
\end{equation}

For the global randomness, the objective function $\Finquad{global}$ is
\begin{multline}
\Finquad{global}(\setof{M^a_{x^*}}{a}, \setof{N^b_{y^*}}{b}, \setof{Z_{ab}}{ab}) = \\
\sum_{ab} \left[M^a_{x^*}N^b_{y^*}(Z_{ab} + Z_{ab}^\dagger + (1-t_i)Z_{ab}^\dagger Z_{ab}) + t_i Z_{ab} Z_{ab}^\dagger\right].
\label{eqn:inner_quad_global}
\end{multline}

Finally, the $i^\mathrm{th}$ term objective function of the blind randomness is given
\begin{multline}
\Finquad{blind}(\setof{M^a_{x^*}}{a}, \setof{N^b_{y^*}}{b}, \setof{Z_{ab}}{ab}) = \\
\sum_{ab} \left[M^a_{x^*}N^b_{y^*}(Z_{ab} + Z_{ab}^\dagger + (1-t_i)Z_{ab}^\dagger Z_{ab}) + t_i N^b_{y^*} Z_{ab} Z_{ab}^\dagger\right].
\label{eqn:inner_quad_blind}
\end{multline}


\section{Entropy accumulation}
\label{appendix:Entropy accumulation}

The entropy accumulation theorem (EAT) provides a way of bounding the smooth min-entropy of a set of $n$ random variables $K_i$ produced by a sequential protocol conditioned on some side information $E_n$ in terms of the conditional von Neumann entropies generated in each round. In this work, we consider the generalized version~\cite{Metger22Oct} of EAT, where the model of side information is relaxed by allowing it to be updated in every round, that is the side information after the $i^\mathrm{th}$ round is $E_i$.
Suppose that we represent the protocol as a sequence of quantum channels $\cM_i\in \CPTP(R_{i-1}E_{i-1}, R_i E_i K_i C_i)$, where for a cryptographic scheme, $K_i$ represents the output in the $i^\mathrm{th}$ round, $E_i$ is the side information leaked to some adversary Eve after the $i^\mathrm{th}$ round, $R_i$ corresponds to some internal system of the device, and $C_i$ is classical information used to determine whether the protocol aborts.  
However, to be able to accumulate entropy from every round, the Markov condition is replaced by a certain non-signaling condition on $\cM_i$: we require that there exists a channel $\cR_i\in\CPTP(E_{i-1}, E_i)$ such that
\begin{equation}
    \Tr_{K_i R_i} \circ \cM_i = \cR_i \circ \Tr_{R_{i-1}}.
\label{eqn:non-signaling_cond}
\end{equation}
Intuitively, this means that all side information about output $K_i$ is already present in $E_i$. The non-signaling condition prevents new information about $K_i$ from being leaked from the $R$-systems at a later round.

To formally state the GEAT with testing, a way to compute the min-entropy conditioned on some classical statistics, we define some relevant concepts below.

\begin{definition}[Frequency distribution]
We denote the frequency of the occurrence of the value $c$ in the sequence $C^n=C_1 C_2 ... C_n$ by 
\begin{equation}
\freq_{C^n}(c)=\frac{\abs{\set{i\in\set{1,2,...,n}\mid C_i=c}}}{n}.
\label{eq:freq_dist}
\end{equation}
\end{definition}

\begin{definition}[EAT channels]

Let $\cM_i\in\CPTP(R_{i-1}E_{i-1}, C_i K_i R_i E_i)$ be a sequence of completely positive trace-preserving (CPTP) maps for $i \in \set{1, \ldots, n}$, where $C_i$ are classical systems with common alphabet $\cC$.
We require that these channels satisfy the following non-signaling condition: defining $\cM'_i  = \Tr_{C_i} \circ \cM_i$, there exists a channel $\cT$ such that 
$
\cM_n \circ \cdots \circ \cM_1 = \cT \circ \cM'_{n} \circ \cdots \circ \cM'_1
$
and $\cT$ has the form
\begin{multline}
\cT(\omega_{A^n E_n}) = \\
\sum_{u \in \cU , v \in \cV} (\Pi_{K^n}^{(u)} \otimes \Pi_{E_n}^{(v)}) \omega_{K^n E_n} (\Pi_{K^n}^{(u)} \\
\otimes \Pi_{E_n}^{(v)}) \otimes \proj{r(u,v)}_{C^n} \,, \label{eqn:measurement_condition}
\end{multline}
where $\set{\Pi_{K^n}^{(u)}}$ and $\set{\Pi_{E_n}^{(v)}}$ are families of mutually orthogonal projectors on $K_i$ and $E_i$, and  $r : \cU \times \cV  \to \cC$ is a deterministic function.  
Intuitively, this condition says that the classical statistics can be reconstructed via projective measurements on systems $K^n$ and $E_n$ at the end of the protocol.
In particular, this requirement is always satisfied if the statistics are computed from classical information contained in $K^n$ and $E_n$. Note that the statistics are still generated in a round-by-round manner; Eq.~(\ref{eqn:measurement_condition}) merely asserts that they could have been reconstructed from the final state.    
\end{definition}

Let $\bP$ be the set of probability distributions on the alphabet $\cC$ of $C_i$, and let $\tE_{i-1}$ be a system isomorphic to $R_{i-1} E_{i-1}$.
For any $q \in \bP$ we define the set of states 
\begin{multline}
\label{eqn:def_sigma}
  \Sigma_i(q) = \\
  \bigset{\nu_{C_i K_i R_i E_i \tE_{i-1}} =  \cM_i(\omega_{R_{i-1} E_{i-1} \tE_{i-1}}) \,|\, \nu_{C_i} = q}  \ ,
\end{multline}
where $\nu_{C_i}$ denotes the probability distribution over $\cC$ with the probabilities given by $\Pr[c] = \bra{c} \nu_{C_i} \ket{c}$. In other words, $\Sigma_i(q)$ is the set of states that can be produced at the output of the channel $\cM_i$ and whose reduced state on $C_i$ is equal to the probability distribution $q$.

\begin{definition}[Min-tradeoff function]
\label{def:tradeoff}
A  function $f: \bP \to \R$, where $\bP$ represents all valid probability distributions, is called a \emph{min-tradeoff function} for $\set{\cM_i}$ if it satisfies 
\begin{align}
f(q) \leq \min_{\nu \in \Sigma_i(q)} H(K_i|E_i \tE_{i-1})_{\nu} \quad \forall i = 1, \dots, n\, .
\label{eqn:min_tradeoff}
\end{align}
Note that if $\Sigma_i(q) = \emptyset$, then $f(q)$ can be chosen arbitrarily.
Our result will depend on some simple properties of the min-tradeoff function, namely the maximum and minimum of $f$, the minimum of $f$ over valid distributions, and the maximum variance of $f$:
\begin{align*}
    \Max(f) &:=\max_{q\in\bP} f(q), \\
    \Min(f) &:=\min_{q\in\bP} f(q), \\
    \MinSig(f) &:=\min_{q:\Sigma(q)\ne \emptyset} f(q), \\
    \VarSig(f) &:= \max_{q:\Sigma(q)\ne \emptyset} \sum_{c\in\cC}q(c)f^2(\delta_c) - (\sum_{c\in\cC}q(c)f(\delta_c))^2,
\end{align*}
and let $\MinSig(f)=-\infty, \VarSig(f)=\infty$ if $\Sigma(q)=\emptyset$.
\end{definition}

\begin{theorem}[GEAT]\cite[Theorem 4.3]{Metger22Oct}
\label{thm:geat_with_testing}
Consider a sequence of channels $\cM_i \in \CPTP(R_{i-1}E_{i-1},\allowbreak C_i K_i R_i E_i)$ for $i \in \{1, \dots, n\}$, where $C_i$ are classical systems with common alphabet $\cC$ and the sequence $\{\cM_i\}$ satisfies Eq.~(\ref{eqn:measurement_condition}) and the non-signaling condition: for each $\cM_i$, there exists a channel $\cR_i \in \CPTP(E_{i-1}, E_i)$ such that $\Tr_{K_i R_i C_i} \circ \cM_i = \cR_i \circ \Tr_{R_{i-1}}$.
Let $\epsilon \in (0,1)$, $\alpha \in (1, 3/2)$, $\Omega \subset \cC^n$, $\rho_{R_0 E_0}$ is the initial state entering the sequence of channels, and $f$ be an affine min-tradeoff function with $h = \min_{c^n \in \Omega} f(\freq_{c^n})$. Then,
\begin{multline}
\Hmineps(K^n | E_n)_{\cM_n \circ \dots \circ \cM_1(\rho_{R_0 E_0})_{|\Omega}} \\
\geq n \, h - n \, \frac{\alpha-1}{2-\alpha} \, \frac{\ln(2)}{2} V^2\\
- \frac{g(\epsilon) + \alpha \log(1/\Pr_{\rho}[\Omega])}{\alpha-1} -  n \, \left( \frac{\alpha-1}{2-\alpha} \right)^2 K'(\alpha)\, ,  \label{eqn:alpha_to_choose}
\end{multline}
where $\Pr[\Omega]$ is the probability of observing event $\Omega$, and with $d_K = \max_i \dim(K_i)$,
\begin{align*}
g(\epsilon) &= - \log(1 - \sqrt{1-\epsilon^2}) \,,\\
V &= \log (2d_A^2+1) + \sqrt{2 + \Var{f}} \,,\\
K'(\alpha) &= \frac{(2-\alpha)^3}{6 (3-2\,\alpha)^3 \ln 2} \, 2^{\frac{\alpha-1}{2-\alpha}(2\log d_{A}  + \Max{f} - \MinSigma{f})} \\
&\quad \cdot \ln^3\left( 2^{2\log d_{A} + \Max{f} - \MinSigma{f}} + e^2 \right).
\end{align*}

\end{theorem}


\section{Construction of min-tradeoff function}
\label{appendix:Construction of min-tradeoff function}

Our construction of the min-tradeoff function is based on the idea of the Lagrange dual function of the optimization problem for the single-round entropy~\cite{Tan21, Tan22}.

The Lagrange dual function for the $i^\mathrm{th}$ SDP inside the summation of Eq.\eqref{eqn:inner_quad_sdp} is
\begin{equation}
\begin{aligned}
g_i(\wexp, \wtol, \ztol) = &\sup_{\lambwv^i, \lambz^i} \min_{M, N, Z^i} F_i(M,N,Z^i) \\
&- \lambw[1]^i \left( \expval{\Gamw(M, N)} - \wexp + \wtol \right) \\
&+ \lambw[2]^i \left( \expval{\Gamw(M, N)} - \wexp - \wtol \right) \\
&- \lambz^i \cdot \left( \ztolv - \expval{\Gamzv(M, N)} \right),
\end{aligned}
\label{eqn:quad_sdp_dual}
\end{equation}
where $\lambwv$ and $\lambz$ are dual variables, $M$ and $N$ denote the sets of POVMs $\setof{M^a_x}{ax}$ and $\setof{N^b_y}{by}$ and $Z^i=\setof{Z^i_{ab}}{ab}$, $\Gamw(M,N)=\sum_{abxy} M^a_x N^b_y$, and $\Gamz[i] = P(a_i b_i |x_i y_i)$ for all $(a_i,b_i|x_i,y_i)\in\cS_{\kappa}$.

Let $g_{\lambwv^i,\lambz^i}$ be the function with given $\lambwv^i$ and $\lambz^i$. By the duality of SDP (Eq.\eqref{eqn:gaussradau_quadrature}), we can bound the von Neumann entropy with $\setof{g_{\lambwv^{i, *},\lambz^{i, *}}}{i}$:
\begin{equation}
\inf_{M,N,Z} H(K|IE') \ge c_m + \sum^{m-1}_{i=1} c_i \Xi^*_i \ge c_m + \sum^{m-1}_{i=1} c_i g_{\lambwv^{i, *},\lambz^{i, *}},
\end{equation}
where $\Xi^*_i$ denote the optimal solution of the corresponding SDP (Eq.\eqref{eqn:inner_quad_sdp}) and $g_{\lambwv^{i, *},\lambz^{i, *}}$ is the function with the optimal dual variables $\lambwv^*$ and $\lambz^*$. The min-tradeoff function can be chosen as
$$f(\wexp, \wtol, \ztol) = c_m + \sum^{m-1}_{i=1} c_i g_{\lambwv^{i, *},\lambz^{i, *}}.$$

To simplify the expression of the min-tradeoff function, we first put $\lambw[1]^i$ and $\lambw[2]^i$ together as $\lambw^i = \lambw[1]^i - \lambw[2]^i$, and define $\lambw:=\sum^{m-1}_{i} c_i\lambw^i$ and $\lambz:=\sum^{m-1}_{i} c_i\lambz^i$. The min-tradeoff function then becomes
\begin{multline}
f(\wexp, \wtol, \ztol) = c_m + \sum^{m-1}_{i=1} c_i \cdot \tilde{f}_i \nonumber\\
\approx \sup_{\setof{\lambwv^i}{i}, \setof{\lambz^i}{i}} \lambw \cdot (\wexp - \wtol) - \lambz\cdot\ztolv \nonumber \quad + \Clwz,
\end{multline}
where
\begin{equation*}
\begin{aligned}
\Clwz &= c_m + \sum^{m-1}_{i=1} c_i \cdot \min_{M, N, Z^i} \Bigl\{ g_i(M,N,Z^i) \\
&- \lambw^i \expval{\Gamw(M, N)} + \lambz^i \cdot \expval{\Gamzv(M, N)} \Bigr\}.
\end{aligned}
\end{equation*}
Here we drop the terms $-2\lambw[2]\wtol$ by the fact that taking out the constraint $\expval{\Gamw(M, N)} \le \wexp + \wtol $ from Eq.\eqref{eqn:inner_quad_sdp} does not significantly affect the optimal solution, that is $\lambw[2]\ll\lambw[1]$ and also $\wtol<\wexp$. Finally, given $\wtol, \ztol \in (0,1)$, and let $\nu=\wexp-\wtol$, we construct the min-tradeoff function as
\begin{equation}
f(\nu) = \lambw\cdot\nu - \lambz\cdot\ztolv + \Clwz.
\end{equation}

To further simplify the last term, $\Clwz$, we consider the maximum winning probability and minimum zero probability among all the $(m-1)$ optimizations:
\begin{equation*}
\begin{aligned}
\gamw^* &= \max_i \gamw^i, \\
\gamz[j]^* &= \min_i \gamz[j]^i,\ j=1,...|\cS_{\kappa}|,
\end{aligned}
\end{equation*}
where $\gamw^i$ is the winning probability $\expval{\Gamw(M, N)}$ evaluated after solving the $i$th optimization problem, and $\gamz[j]^i = P(a'b'|x'y')=\expval{\Gamz[j](M,N)}$ denotes the probability of the input-output combination $(a',b',x',y')\in\cS_{\kappa}$ evaluated after solving the $i$th optimization problem. Using the following facts:
\begin{equation*}
\begin{aligned}
- \sum^{m-1}_{i=1} c_i \lambw^i \gamw^i &\ge - \sum^{m-1}_{i=1} c_i \lambw^i \gamw^* = - \lambw\gamw^*, \\
\sum^{m-1}_{i=1} c_i \lambz^i \cdot \gamzv^i &\ge \sum^{m-1}_{i=1} c_i \lambz^i \cdot \gamzv^* = \lambz \cdot \gamzv^*. 
\end{aligned}
\end{equation*}
We obtain a shortened expression for $\Clwz$:
\begin{equation}
\Clwz =  - \lambw\gamw^* + \lambz \cdot \gamzv^* + c_m + \sum^{m-1}_{i=1} c_i \cdot \Xi_i^*.
\end{equation}


\section{Cross-over min-tradeoff function}
\label{appendix:Cross-over min-tradeoff function}

To properly account for how the statistics observed would extend to the non-testing (generation) rounds, we need to modify the min-tradeoff function to the so-called ``crossover'' min-tradeoff function~\cite{Dupuis18} that weighs the contributions from both testing and generation rounds.
By the requirement that the properties of the min-tradeoff function have to hold for the testing-round channel, we can construct the crossover min-tradeoff function associated with the testing ratio $\gamma$~\cite{Dupuis18, Liu21, Tan22}.

\begin{lemma}[Crossover min-tradeoff function~\cite{Liu21}]
\label{lem:cross_mtf}
Let $\cC=\{0,1\}$,  $\cC'=\cC\cup \{\perp\}$, and $q\in \bP(\cC)$, $q'\in \bP(\cC')$, where $q'(\perp)=(1-\gamma)$ and $q'(c)=\gamma q(c)$. The crossover min-tradeoff function corresponding to the infrequent-sampling channel $\cM_i$ with sampling (testing) probability $\gamma$ is
\begin{align}
\fgamma(\delta_c) = \begin{cases}
    \frac{1}{\gamma}f(\delta_c)+(1-\frac{1}{\gamma})\fperp & c\in\cC/\{\perp\}, \\
    \fperp & c = \perp, \label{eqn:cross_mtf}
\end{cases}
\end{align}
which satisfies
\begin{equation}
\begin{aligned}
&\fgamma(q') = \sum_{c\in\cC'} q'(c)\fgamma(\delta_c) \\
&= (1-\gamma)\fperp + \gamma\sum_{c\neq \perp} q(\delta_c)(\frac{1}{\gamma}f(\delta_c)+(1-\frac{1}{\gamma})\fperp) = f(q),
\label{eqn:cross_mtf_eq_cond}
\end{aligned}
\end{equation}
where $\fperp\in[\Min(f), \Max(f)]$ can be chosen so that it optimizes the extractable rate.
\end{lemma}

\begin{lemma}[Properties of the crossover min-tradeoff function]
\label{lem:prop_cross_mtf}
The properties of the crossover min-tradeoff function defined in Lemma~\ref{lem:cross_mtf} can be parameterized as follows.
\begin{equation}
\begin{aligned}
\Max(\fgamma) &= (1-\frac{1}{\gamma})\lambda\nu' + \frac{\lambda}{\gamma} + \Clamb, \\
\MinSigamma(\fgamma) &\ge (1-w_Q)\lambda + \Clamb, \\
\VarSigamma(\fgamma) &= \max_{\nu\in [1-w_Q,w_Q]} -\lambda^2(\nu-\nu_0)^2+D,
\label{eqn:prop_cross_mtf}
\end{aligned}
\end{equation}
where the set of state $\Sigamma(q)=\{\omega_{A_iC_iR_iE_i\tilde{E}_{i-1}}=\cM_i(\omega_{R_{i-1}E_{i-1}\tilde{E}_{i-1}}) \mid \omega_{C_i}(c=\perp)=(1-\gamma),\ \omega_{C_i}(c\neq\perp)=\gamma q(c)\}$, 
$w_Q=\frac{2+\sqrt{2}}{4}$,
$\nu_0=\frac{1}{2\gamma}+(1-\frac{1}{\gamma})\nu'$,
and $D=\frac{\lambda^2}{4\gamma^2} + \frac{1}{\gamma}(1-\frac{1}{\gamma})\lambda^2(1-\nu')\nu'$ with $\nu'\in[0,1]$.
\end{lemma}

\begin{proof}
Following the steps shown in~\cite{Liu21}, we first write the properties of the crossover min-tradeoff function as follows.
\begin{equation}
\begin{aligned}
\Max(\fgamma) &=\max \left\{\frac{1}{\gamma}\Max(f)+(1-\frac{1}{\gamma})\fperp, \fperp \right\}, \\
\MinSigamma(\fgamma) &=\MinSig(f), \\
\VarSigamma(\fgamma) &= \max_{q:\Sigamma(q)\ne \emptyset} \gamma\sum_{c\neq\perp}q(c)\fgamma^2(\delta_c)+(1-\gamma)\fperp^2-f^2(q). \label{eqn:prop_cross_mtf_proof}
\end{aligned}
\end{equation}
Since $\fperp\in[\Min(f), \Max(f)]$, we redefine   $\fperp=\lambda\nu'+\Clamb$ with $\nu'\in[0,1]$. We can simplify the first quantity.
\begin{equation}
\begin{aligned}
\Max(\fgamma) &=\max \left\{\frac{1}{\gamma}(\lambda+\Clamb)+(1-\frac{1}{\gamma})\fperp, \fperp \right\} \\
&=\max \left\{\fperp+\frac{1}{\gamma}(1-\nu')\lambda, \fperp \right\} \\
&= \fperp+\frac{1}{\gamma}(1-\nu')\lambda \\
&= (1-\frac{1}{\gamma})\lambda\nu' + \frac{\lambda}{\gamma} + \Clamb.
\end{aligned}
\label{eqn:max_cross_mtf_proof}
\end{equation}
The third equation holds by the fact that $\gamma >0, \lambda >0$ and $1-\nu'\ge 0$.

For the second equation in Eq.\eqref{eqn:prop_cross_mtf}, the minimum is taken over the set $\Sigamma$, that the equality holds according to Eq.\eqref{eqn:cross_mtf_eq_cond}.

To prove the last equation in Eq.\eqref{eqn:prop_cross_mtf}, first we compute the expression with a summation:
\begin{align}
&S = \sum_{c\neq\perp} q(c)\fgamma^2(\delta_c) = \nu \fgamma^2(\delta_1) + (1-\nu) \fgamma^2(\delta_0) \nonumber\\
&= \nu \left[\fperp + \frac{1}{\gamma}(1-\nu')\lambda\right]^2 + (1-\nu) \left(\fperp - \frac{1}{\gamma}\nu'\lambda\right)^2 \nonumber\\
&= \nu\left[2\fperp + \frac{1}{\gamma}(1-2\nu')\lambda\right] \cdot \frac{\lambda}{\gamma} + \left(\fperp - \frac{1}{\gamma}\nu'\lambda\right)^2.
\label{eqn:var_cross_mtf_proof}
\end{align}
The second line holds by the fact that
\begin{equation*}
\begin{aligned}
\fgamma(\delta_1) &= \frac{1}{\gamma}(\lambda+\Clamb) + (1-\frac{1}{\gamma})\fperp \\
&= \fperp + \frac{1}{\gamma}(1-\nu')\lambda, \\
\fgamma(\delta_0) &= \frac{1}{\gamma}\Clamb + (1-\frac{1}{\gamma})\fperp \\
&= \fperp - \frac{1}{\gamma}\nu'\lambda.
\end{aligned}
\end{equation*}
For a distribution in the set $\Sigamma$, the winning probability satisfies $\nu\in[1-w_Q,w_Q]$, where $w_Q=(2+\sqrt{2})/4$. Put Eq.\eqref{eqn:var_cross_mtf_proof} back to the third equation in Eq.\eqref{eqn:prop_cross_mtf_proof}, we have
\begin{equation}
\begin{aligned}
\VarSigamma(\fgamma) =& \max_{\nu\in [1-w_Q,w_Q]} \lambda\nu\left[2\fperp + \frac{1}{\gamma}(1-2\nu')\lambda\right] \\ &+ \gamma \left(\fperp - \frac{1}{\gamma}\nu'\lambda\right)^2 + (1-\gamma)f^2_\perp - (\lambda\nu+\Clamb)^2 \\
=& \max_{\nu\in [1-w_Q,w_Q]} -\lambda^2\nu^2 + \lambda^2\left[2\nu'+\frac{1}{\gamma}(1-2\nu')\right]\nu \\
&+ \left(f^2_\perp - 2\nu'\lambda \fperp + \frac{\nu'^2\lambda^2}{\gamma} - \Clamb^2\right).
\end{aligned}
\label{eqn:var_cross_mtf_proof-3}
\end{equation}

Eq.\eqref{eqn:var_cross_mtf_proof-3} is a quadratic function of $\nu$, and therefore we can rewrite it as
\begin{equation}
\VarSigamma(\fgamma) = \max_{\nu\in [1-w_Q,w_Q]} -\lambda^2(\nu-\nu_0)^2+D,
\end{equation}
where $\nu_0=\nu'+\frac{1}{\gamma}(\frac{1}{2}-\nu')=\frac{1}{2\gamma}+(1-\frac{1}{\gamma})\nu'$ and
\begin{align*}
D =& f^2_\perp - 2\lambda\nu' \fperp + \frac{1}{\gamma} \lambda^2\nu'^2 - \Clamb^2 + \lambda^2\nu^2_0 \\
=& (\lambda^2\nu'^2 + 2\lambda \Clamb\nu' +\Clamb^2) -2\lambda\nu'(\lambda\nu'+\Clamb) \\
&+ \frac{1}{\gamma} \lambda^2\nu'^2 - \Clamb^2 + \lambda^2\left[\frac{1}{2\gamma}+(1-\frac{1}{\gamma})\nu'\right]^2 \\
=& - \left(1-\frac{1}{\gamma}\right)\lambda^2\nu'^2 \\
&+ \lambda^2 \left[\frac{1}{4\gamma^2}+\frac{1}{\gamma}(1-\frac{1}{\gamma})\nu' + (1-\frac{1}{\gamma})^2\nu'^2\right] \\
=& \frac{\lambda^2}{4\gamma^2} + \frac{1}{\gamma}(1-\frac{1}{\gamma})\lambda^2(1-\nu')\nu'.
\end{align*}
The last equation holds by combining $-(1-\frac{1}{\gamma})\lambda^2\nu'^2$ and $(1-\frac{1}{\gamma})^2\lambda^2\nu'^2$ together as $-\frac{1}{\gamma}(1-\frac{1}{\gamma})\lambda^2\nu'^2$.
\end{proof}


\section{Protocol security and finite rate lower bound}
\label{appendix:Protocol security and finite rate lower bound}

In Protocol~\ref{protocol:DI randomness generation}, given $\wexp\in[w_C, w_Q]$, $\gamma\in(0,1]$, $\nu'\in[0,1]$, and $\epsilon, \wtol, \ztol, \beta\in(0,1)$, the lower bound on the smooth min-entropy can be derived by directly applying the GEAT (Theorem~\ref{thm:geat_with_testing}.)
\begin{equation}
\Hmineps(K^n|I^nE'_{n}) \ge h(\wexp-\wtol) - \Delta.
\label{eqn:smooth_min_entropy_bound}
\end{equation}
$\Delta$ is given below
\begin{align*}
\Delta =& \frac{\ln 2}{2}\beta \left[\log 9 + V\right]^2 \\
&- \frac{1}{n}\bigl[ \frac{1+\beta}{\beta} \log(1-\sqrt{1-\epsilon^2}) + \frac{1+2\beta}{\beta}\log\Pr[\Omega]\bigr] \\
&- \frac{1}{6\ln 2} \frac{\beta^2}{(1-\beta)^3} \zeta^\beta \ln^3(\zeta+e^2), 
\end{align*}
where $\beta=\frac{\alpha-1}{2-\alpha}=\frac{1}{2-\alpha}-1\in(0,1)$,
$V=\sqrt{2+D-\min_{\nu\in[1-w_Q,w_Q]}\lambda^2(\nu-\nu_0)^2}$, $\gamma_0=\frac{1-\gamma}{\gamma}$,
$\nu_0=\frac{1}{2\gamma}-\gamma_0\nu'$, $D=\frac{\lambda^2}{4\gamma^2} - \frac{\gamma_0}{\gamma} \lambda^2(1-\nu')\nu'$, $\zeta=2^{2+\lambda (\gamma_0 (1-\nu') + w_Q)}$, and $\Hsh$ is Shannon entropy. For the CHSH game, $w_C=0.75$ and $w_Q=(2+\sqrt{2})/4$.

\begin{proof}
To map the systems described in Theorem~\ref{thm:geat_with_testing}, we first combine all the statistical checks in the $i^\mathrm{th}$ round $C_i=(C_{\omega,i}, \vec{C}_{3b, i})$, where $\vec{C}_{3b, i} = (C_{3b, i}[00|00], C_{3b, i}[11|00], C_{3b, i}[10|11])$. The system records the adversary's total side information as $E_i=I^i T^i E'^i$, where $I_i$ contains all the classical side information except the round mark for testing or generation $T^i$. Denote the devices' internal memory as $R_i$.
The non-signaling condition can be verified by constructing the map $\cR_i\in\CPTP(E_{i-1}, E_i)$ with the isometries creating Eve's classical side information $I_i$.

Then by employing Theorem~\ref{thm:geat_with_testing} with the properties of the min-tradeoff function shown in Lemma~\ref{lem:prop_cross_mtf}, we can derive Eq.\eqref{eqn:smooth_min_entropy_bound}.
\end{proof}

To prove the security of the protocol, we need concentration inequality to upper-bound the completeness. For convenience, we choose Hoeffding's inequality, while in principle other concentration inequalities, such as Serfling's inequality~\cite{Zhou17}, Bernstein's inequality~\cite{Metger22Mar}, and Chernoff's bound~\cite{Tan22}, can also be used to derive the completeness bound.

\textbf{Hoeffding's inequality}~\cite[Proposition 1.2]{Bardenet15}
\label{lem:hoeffding}
Consider a finite sequence $x^N=x_1x_2...x_N$ with the mean $\mu=\frac{1}{N}\sum^N_{i=1} x_i$. If the $n$ samples drawn without replacement from the sequence $x^N$ is denoted by $X^n=X_1X_2...X_n$, then given $\epsilon > 0$, we have
\begin{equation}
\Pr[\frac{1}{n}\sum^n_{i=1} X_i-\mu \ge \epsilon] \le \exp\left(-\frac{2n{\epsilon}^2}{{(b-a)}^2}\right), \\
\label{eqn:hoeffding-1}
\end{equation}
\begin{equation}
\Pr[\mu - \frac{1}{n}\sum^n_{i=1} X_i \ge \epsilon] \le \exp\left(-\frac{2n{\epsilon}^2}{{(b-a)}^2}\right),
\label{eqn:hoeffding-2}
\end{equation}
where $a=\min x_i$ and $b=\max x_i$.

\begin{theorem}[Security of Protocol~\ref{protocol:DI randomness generation}]
Protocol~\ref{protocol:DI randomness generation}
Given the protocol parameters, $(n, \gamma, \wexp, \wtol, \ztol, \epsilon, \epsext)$, Protocol~\ref{protocol:DI randomness generation} is $\epsound$-sound and $\epscomp$-complete, where $\epsound = \epsext+2\epsilon$ and $\epscomp = 1 - (1-e^{-2\wtol^2  n})(1-e^{-2\ztol'^2 n})^{n_\mathrm{zero}}$, $n_\mathrm{zero}$ is the number of zero-probability constraints depending on the choice of the classes $\kappa$.
\end{theorem}

\begin{proof}
Consider the event of non-aborting $\Omega=[\text{NonAbort}]=[\text{Pass win-probability check}]\cup[\text{Pass zero-probability checks}]$. If any one of the checks does not pass, then set the key register $\rho_{K^n} = (\proj{\perp})^{\otimes n}$, where $\bra{k^n}(\ket{\perp})^{\otimes n} = 0$ for all $k^n \in \set{0,1}^n$. The final state of Protocol~\ref{protocol:DI randomness generation} can be written as
\begin{equation}
\rho_{K^n E} = \Pr[\Omega] \rho_{K^n E_{|\Omega}} + (1-\Pr[\Omega]) (\proj{\perp})^{\otimes n} \otimes \rho_{E_{|\neg\Omega}}.
\end{equation}

If the smooth min-entropy $\Hmineps(K^n|E)_{\rho_{K^n E_{|\Omega}}}$ is lower-bounded by $\kext$, the soundness of Protocol~\ref{protocol:DI randomness generation} is guaranteed by Lemma~\ref{def:quantum-proof extractor} after applying a $(\kext,\epsext)$-strong extractor
\begin{equation}
\frac{1}{2}\|\rho_{K^n E_{|\Omega}}-\mu_K\otimes\rho_{E_{|\Omega}}\|_1 \le \epsext+2\epsilon,
\end{equation}    
where $\mu_K=\frac{1}{|\cK|}\id$ is the maximally mixed state on $\Hilbert_K$. And by the definition of soundness (Eq.\eqref{eqn:soundness_def}) and the fact that $\Pr[\Omega]\le 1$, we can choose $\epsound=\epsext+2\epsilon$.

We derive the completeness by bounding the aborting probability due to the i) check of the score $C_\omega^n$ and ii) check of the zero-probability constraints $\vec{C}_{3b}^n$, taking $\kappa=3b$ for example~\footnote{For other classes, different the number of zero-probability constraints would approximately give a different scaling parameter for the zero-probability constraint related aborting probability.}.

By Heoffding's inequality~\eqref{eqn:hoeffding-2}, let $c^N$ be the results of $N$ rounds of the nonlocal game: $c_i=1$ ($c_i=0$) if the players win (lose) the $i^\mathrm{th}$ round of the game. Take large enough $N$ such that the statistics of $c^N$ reflect nearly the true probability distribution, and $\wexp=\mathbb{E}[c]$ is the expectation value of the winning probability. If in the experiment, we only run $n$ rounds of the nonlocal game and the results of that $n$ rounds as $C^n$. The aborting probability can be upper-bounded by
\begin{multline}
\Pr[\mathrm{aborted\ by\ score\ check}] = \Pr\left[\mu-\frac{1}{n}\sum_i C_i > \wtol\right] \\
\le \epscomp \le e^{-2\wtol^2  n},
\label{eqn:win_completeness}
\end{multline}
where $\frac{1}{n}\sum_i C_i$ is the observed winning probability.

Similarly, let $c^N$ be the occurrence of one of the input-output combinations $(a,b,x,y)\in\cS_{\kappa}$ in $N$ rounds of the nonlocal game. Take large enough $N$ such that the statistics of $c^N$ reflect nearly the true probability distribution, and $\mu=\mathbb{E}[c]$ is close to the expectation value of the probability, $P(a,b|x,y)$. If in the experiment, we only run $n$ rounds of the nonlocal game, and denote the occurrence of $(a,b,x,y)$ in that $n$ rounds as $C^n$. Then, the aborting probability due to the zero probability check can be upper-bounded by applying Hoeffding's inequality~\eqref{eqn:hoeffding-1},
\begin{multline*}
\Pr[\mathrm{aborted\ by\ one\ zero\ check}] \\
= \Pr[\frac{1}{n}\sum_i C_i -\mu \ge \ztol'] \le e^{-2\ztol'^2 n},
\end{multline*}
where the equality $\mu+\ztol'=\ztol$ holds and $\ztol$ is the zero tolerance level for Protocol~\ref{protocol:DI randomness generation}. Since $\mu >0$ and $\ztol' >0$, the value, $\ztol'$  is strictly smaller than $\ztol$. Since there are three zero-probability constraints in class 3b, the effective completeness for all the zero-probability checks can be bounded by
\begin{equation}
\epscz \le 1 - (1-e^{-2\ztol'^2 n})^3.
\label{eqn:zero_prob_completeness}
\end{equation}

Combine Eq.\eqref{eqn:win_completeness} and Eq.\eqref{eqn:zero_prob_completeness} together with the fact that for two events $\Omega_A$ and $\Omega_B$,
$\Pr[\Omega_A \cap \Omega_B]\ge\Pr[\Omega_A]\Pr[\Omega_B]$, we have the non-aborting probability conditioned on the honest behavior $\Pr[NonAbort|Honest]\equiv\Pr[\Omega_{|Hon}]$:
\begin{equation}
\begin{aligned}
\Pr[\Omega_{|Hon}]&=\Pr[(\mathrm{pass\ win\ check})\cap(\mathrm{pass\ zero\ check})] \\
&\ge \Pr[\mathrm{pass\ win\ check}]\Pr[\mathrm{pass\ zero\ check}] \\
&=(1-\Pr[\mathrm{aborted\ by\ win\ check}])\\
&\quad\cdot (1-\Pr[\mathrm{aborted\ by\ zero\ check}]) \\
&\ge (1-\epscw)(1-\epscz) \\
&\ge (1-e^{-2\wtol^2  n})(1-e^{-2\ztol'^2 n})^3.
\end{aligned}
\label{eqn:non_aborting_prob}
\end{equation}
We obtain the aborting probability of the protocol under the honest implementation:
\begin{equation}
\begin{aligned}
\Pr[\mathrm{Abort}|\mathrm{Honest}] &= 1-\Pr[\Omega_{|Hon}] \\
&\le 1 - (1-e^{-2\wtol^2  n})(1-e^{-2\ztol'^2 n})^3.
\label{eqn:aborting_prob}
\end{aligned}
\end{equation}
\end{proof}


\section{Numerical implementation and results}
\label{appendix:Numerical results}

Fix the protocol parameter $\wtol=10^{-4}$ and $\epsilon=10^{-12}$. Two variables still need to be optimized, i) $\nu'$ from the construction of crossover min-tradeoff function and ii) $\beta=$ which is directly related to the order $\alpha$ of the $\alpha$-Renyi entropy. For each fixed number of rounds $N$, we scan through the valid ranges, $\nu'\in[0,1]$ and $\beta\in(0,1)$, to find the values that lead to the highest rates. As an illustration of this procedure, we plot the resulting heat map in Fig~\ref{fig:beta_scan} to show the optimization of $\beta$.

\begin{figure}[ht!]
    \centering
    \includegraphics[width=.9\linewidth]{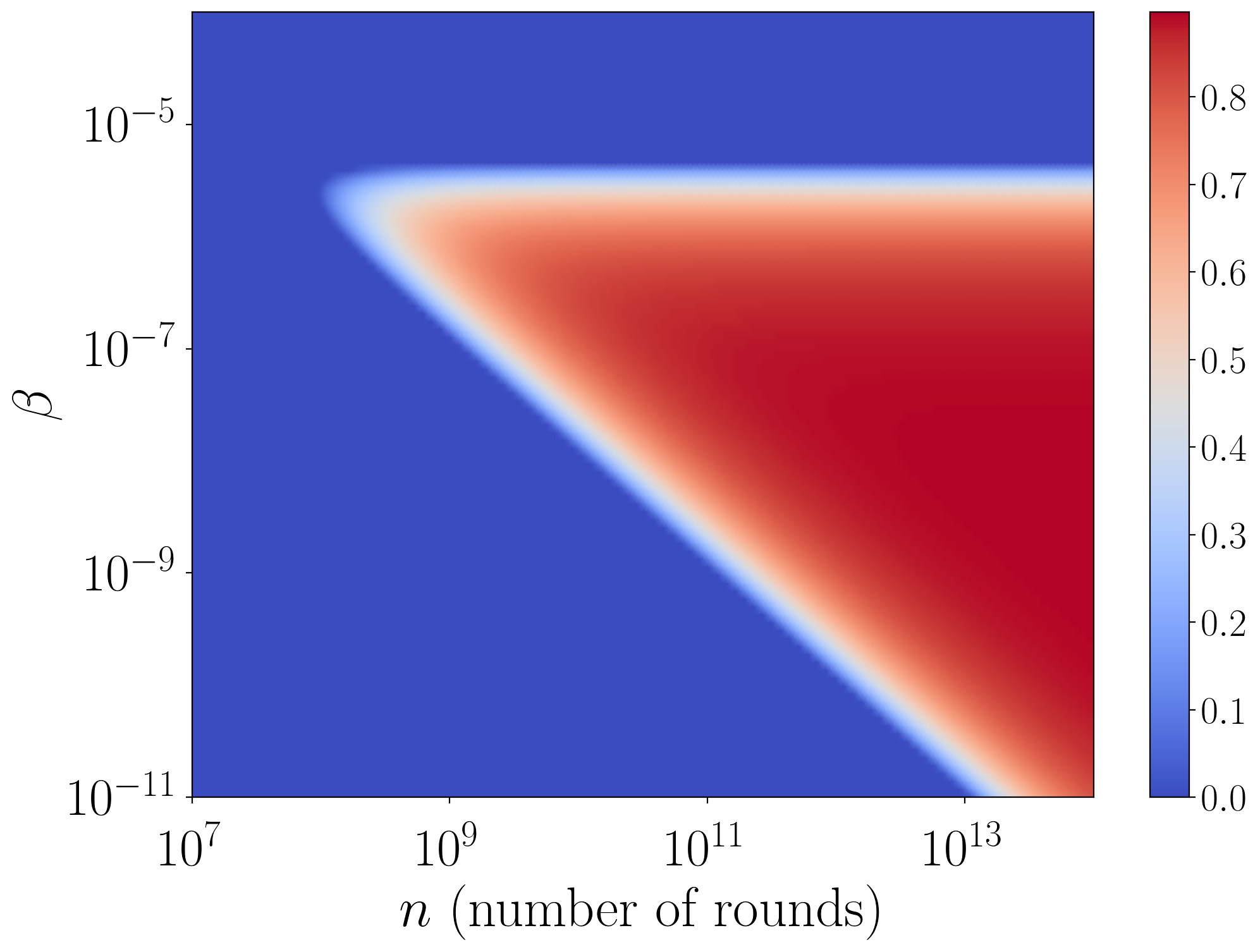}
    \caption{Scanning $\beta$ at optimal $\nu'$. The hotter color indicates a higher rate.}
    \label{fig:beta_scan}
\end{figure}

We use Eq.\eqref{eqn:non_aborting_prob} to upper-bound the non-aborting probability $\Pr[\Omega]$ instead of $1-\epscw-\epscz$, and taking $\ztol'=\ztol/2$ for convenience.

The comparison of Protocol~\ref{protocol:DI randomness generation}, CHSH, and the $\delta$-family Bell inequality~\cite{Wooltorton22} for single-party and blind randomness with maximal winning probability are shown in Fig.~\ref{fig:finite_rate}.

\end{document}